\documentclass[submission,Phys]{SciPost}

\usepackage[T1]{fontenc}
\usepackage[bitstream-charter]{mathdesign}
\usepackage{amsthm}
\usepackage{tikz}
\usepackage{booktabs}
\usepackage{float}

\DeclareSymbolFont{usualmathcal}{OMS}{cmsy}{m}{n}
\DeclareSymbolFontAlphabet{\mathcal}{usualmathcal}

\newtheorem{theorem}{Theorem}[section]
\newtheorem{corollary}[theorem]{Corollary}
\newtheorem{lemma}[theorem]{Lemma}
\newtheorem{proposition}[theorem]{Proposition}
\theoremstyle{definition}
\newtheorem{definition}[theorem]{Definition}
\numberwithin{equation}{section}

\hypersetup{
  colorlinks,
  linkcolor={red!50!black},
  citecolor={blue!50!black},
  urlcolor={blue!80!black},
  pdftitle={Monotones from multi-invariants: the Coxeter classification},
  pdfauthor={Abhijit Gadde and Shraiyance Jain}
}

\begin{document}

\begin{center}
  {\Large\textbf{Monotones from multi-invariants:\\
  the Coxeter classification}}
\end{center}

\begin{center}
  Abhijit Gadde\textsuperscript{1$\star$} and
  Shraiyance Jain\textsuperscript{1}
\end{center}

\begin{center}
  {\bf 1} Department of Theoretical Physics,
  Tata Institute of Fundamental Research, Mumbai 400005, India\\
  ${}^\star${\small\sf abhijit@theory.tifr.res.in}
\end{center}

\begin{center}
  \today
\end{center}

\section*{Abstract}
{\bf
Multi-invariants are multiplicative polynomial invariants of multipartite
pure states under local unitaries. They yield entanglement monotones when their
contraction graphs satisfy edge-convexity. We classify all connected
edge-convex multi-invariants. Up to duplication of entire edge-colour
families, their contraction graphs are precisely the standard edge-labelled
Cayley graphs of finite Coxeter systems. Necessity
follows by identifying edge-reflecting graphs with Coxeter Cayley graphs.
For sufficiency, we prove vertex-convexity of parabolic quotients and lift
their matrices by induction. This covers all classical and exceptional
families. For the largest quotients, $E_8/D_7$ and $H_4/A_3$, wall
symmetry reduces the matrices to exact positivity tests over $\mathbb Q$ and
$\mathbb Q(\sqrt5)$.
}

\vspace{10pt}
\noindent\rule{\textwidth}{1pt}
\tableofcontents\thispagestyle{fancy}
\noindent\rule{\textwidth}{1pt}
\vspace{10pt}

\section{Introduction and main result}

Entanglement is a resource for quantum-information tasks performed by
separated parties. The free transformations are local operations and
classical communication (LOCC): each party acts only on its own subsystem and
coordinates with the others by exchanging classical messages. A
pure-state entanglement monotone (PSEM) is a real-valued function \(\nu\) of
normalized pure states that vanishes on fully factorized states and cannot
increase on average under LOCC. Suppose a local operation maps
the initial state \(\lvert\psi\rangle\) to the ensemble
\(\{p_i,\lvert\psi_i\rangle\}\), where \(p_i\) is the probability of the
outcome \(\lvert\psi_i\rangle\). Monotonicity requires
\begin{equation}
  \nu(\lvert\psi\rangle)
  \geq \sum_i p_i\,\nu(\lvert\psi_i\rangle).
  \label{eq:psem}
\end{equation}
Such quantities measure entanglement and constrain the probability of
transforming one state into another by LOCC
\cite{Vidal2000,Horodecki2009}.

Polynomial invariants under local unitary transformations provide
candidates for \(\nu\).
A \emph{multi-invariant} is a multiplicative polynomial formed by contracting
equal numbers of copies of a multipartite state and its complex conjugate
\cite{GaddeHarperKrishna2024}. Here multiplicative means that, for a tensor
product taken party by party,
\begin{equation}
  \mathcal Z(\lvert\psi\rangle\otimes\lvert\phi\rangle)
  =
  \mathcal Z(\lvert\psi\rangle)\,
  \mathcal Z(\lvert\phi\rangle).
\end{equation}
Each contraction pattern is encoded by
a bipartite graph in which state tensors are white vertices, conjugate tensors
are black vertices, and each party supplies one edge label. We call this
contraction graph a \(\psi\)-graph. Following
Refs.~\cite{GaddeJainKulkarni2024,GaddeHarperKrishna2024}, we use the same
calligraphic letter \(\mathcal Z\) for a multi-invariant and its
\(\psi\)-graph. If \(n_{\mathcal Z}\) is the number of white vertices of
\(\mathcal Z\), then for a nonnegative multi-invariant we define its
normalized form by
\begin{equation}
  \widehat{\mathcal Z}:=\mathcal Z^{1/n_{\mathcal Z}}.
  \label{eq:normalized-invariant}
\end{equation}
The graph is called \emph{edge-convex} when its reflecting cuts admit a
compatible family of positive-semidefinite matrices, meaning matrices with
nonnegative quadratic forms. Here a reflecting cut is a set of edges fixed by
an involutive graph symmetry that exchanges the two sides of the cut.
Ref.~\cite{GaddeJainKulkarni2024} proves that
\begin{equation}
  \mathcal Z\ \text{edge-convex}
  \quad\Longrightarrow\quad
  \widehat\nu(\mathcal Z):=1-\widehat{\mathcal Z}\ \text{is a PSEM}.
  \label{eq:edge-convex-implies-psem}
\end{equation}
We restrict the classification to connected contraction graphs.

The matrices attached to the reflecting cuts must jointly satisfy one
normalization condition for every pair of equally labelled edges. For each
such pair, the condition requires at least one reflecting cut that separates
the two edges. A graph with this property is called
\emph{edge-reflecting}. The graph-theoretic characterization in
Ref.~\cite{Marc2017} then implies that connected edge-reflecting
\(\psi\)-graphs, after merging globally parallel colour families, are
precisely the standard Cayley graphs of finite Coxeter groups. Here two
colours are globally parallel if their edges have the same endpoints at
every vertex. Physically, merging globally parallel colours means grouping the corresponding parties into a single party, whose Hilbert space is the tensor product of their local Hilbert spaces. The same state is then viewed as a state of fewer parties. To prove sufficiency, we must construct
positive-semidefinite matrices satisfying the simultaneous
normalization conditions required by edge-convexity.

A \emph{finite Coxeter system} \((W,S)\) consists of a
finite group \(W\) with the presentation
\begin{equation}
  W=\bigl\langle S\,\big|\,(st)^{m_{st}}=1\bigr\rangle,
  \qquad m_{ss}=1,\quad
  m_{st}=m_{ts}\geq2\ \text{for }s\neq t,
  \label{eq:coxeter-presentation}
\end{equation}
where the \(m_{st}\) are integers. In particular, each \(s\in S\) is an
involution, meaning that \(s^2=1\). We write
\(\operatorname{Cay}(W,S)\) for the edge-labelled Cayley graph whose vertices
are the elements of \(W\), with an edge labelled by \(s\) between \(w\) and
\(ws\). It becomes a $\psi$-graph by colouring a vertex white or black
according to the parity of the length of a shortest word for $w$ in the
generators $S$. The \emph{Coxeter diagram} has vertex set $S$ and joins
distinct $s,t$ when $m_{st}\geq3$, with the edge labelled by $m_{st}$ when
$m_{st}>3$. A system is \emph{irreducible} when this diagram is connected.
We use the standard conventions for finite Coxeter systems, root systems, and
parabolic subgroups in Ref.~\cite{Humphreys1990}. Coxeter-labelled
multi-invariants also appear in the study of bulk replica symmetry
\cite{GaddeHarperKrishna2024}. Our main result is the following classification.

\begin{theorem}[Coxeter classification]
  \label{thm:classification}
  A connected \(\psi\)-graph \(\mathcal Z\) is edge-convex if and only if
  it is obtained from the standard Cayley graph of a finite Coxeter system
  \((W,S)\) by replacing every edge of each generator colour by the same
  nonempty family of distinctly labelled parallel edges. In particular,
  for a simple graph the classification is
  \begin{equation}
    \mathcal Z\cong\operatorname{Cay}(W,S),
    \label{eq:classification}
  \end{equation}
  where the isomorphism preserves the edge labels.
\end{theorem}
\noindent
Combining Theorem~\ref{thm:classification} with
Eq.~\eqref{eq:edge-convex-implies-psem} gives the following family of monotones.

\begin{corollary}
  \label{cor:coxeter-monotones}
  For every finite Coxeter system \((W,S)\), let \(\mathcal Z(W,S)\) be the
  multi-invariant associated with \(\operatorname{Cay}(W,S)\). Then
  \begin{equation}
    \widehat\nu\bigl(\mathcal Z(W,S)\bigr)
    :=1-\widehat{\mathcal Z}(W,S)
  \end{equation}
  is a PSEM.
\end{corollary}
\noindent
The theorem classifies the multi-invariants satisfying the edge-convexity
criterion.
For the Cayley graph of a finite Coxeter system $(W,S)$, each generator
labels one party. The graph has $|W|$ vertices, half representing copies of
the state and half its complex conjugate, so the number of replicas is
$|W|/2$. Thus $H_3$ gives a three-party invariant with $60$ replicas, and
$H_4$ a four-party invariant with $7200$ replicas. Their normalized
complements are entanglement monotones for arbitrary finite local dimensions.

The sufficiency proof uses smaller coset spaces to construct the matrices
on the full Coxeter graph. For a subset \(I\subset S\), let
\(W_I:=\langle I\rangle\) be the standard parabolic subgroup generated by
\(I\), and let \(X_I:=W/W_I\) be its coset space. To establish
\emph{vertex-convexity} of \(X_I\), we assign a positive-semidefinite
matrix to every reflection hyperplane, called a reflecting wall, so that the
matrix entries sum to one for each pair of vertices on opposite sides.
The term \emph{reflecting wall} refers to this hyperplane in the quotient,
whereas a reflecting cut is a set of graph edges.
Parabolic induction establishes the
implication
\begin{equation}
  \begin{gathered}
    \operatorname{Cay}(W_I,I)\ \text{is edge-convex},\\
    X_I\ \text{is vertex-convex}
  \end{gathered}
  \quad\Longrightarrow\quad
  \operatorname{Cay}(W,S)\ \text{is edge-convex for the labels in }I.
  \label{eq:parabolic-induction-intro}
\end{equation}
The full graph decomposes into parabolic fibres indexed by \(X_I\).
For a fixed label \(s\in I\), let \(Q^{(\alpha)}\) be the matrix
associated with a reflecting wall \(\alpha\) of the quotient \(X_I\), and let
\(J_s\) be the all-ones matrix indexed by the \(s\)-edges in one parabolic
fibre. The corresponding matrix on the full Cayley graph is
\begin{equation}
  Q^{(\alpha)}\otimes J_s.
\end{equation}
This tensor-product lift preserves positive semidefiniteness and the
normalization conditions. Choosing
parabolics whose generating sets cover \(S\) therefore proves edge-convexity
of the full group.

For the classical Coxeter families, the required quotients are familiar
finite geometries: simplices, hypercubes, cross-polytopes (orthoplexes), and
half-hypercubes (demi-hypercubes). Most exceptional quotients are finite
orbits of roots or weights of finite-dimensional representations; their
matrices $Q^{(\alpha)}$ depend on only a few invariant inner products. Only
two quotients require large positivity checks. The wall-preserving symmetry
commutes with each matrix, allowing us to block-diagonalize it using
irreducible representations. An irreducible representation has no nonzero
proper invariant subspace.

For \(E_8/D_7\), the positive side of a wall has
\(702\) vertices, and the symmetry preserving the wall and its positive
side has type $E_7$.
On five irreducible \(W(E_7)\)-modules of dimensions
\[
  15,\ 21,\ 105,\ 120,\ 189,
\]
the matrix acts as a scalar multiple of the identity. Each of the
irreducible modules of dimensions
\[
  1,\ 7,\ 27,\ 35,\ 56
\]
occurs twice, once in each height layer, where a layer consists of vertices
with the same inner product with the root normal to the wall. On the corresponding two-copy
component, the matrix therefore has the form
\[
  M_d\otimes I_d,
  \qquad M_d\in\mathbb R^{2\times2},
\]
where \(I_d\) is the identity matrix on the \(d\)-dimensional irreducible
module. Thus positivity of the original matrix reduces to positivity of five
scalars and five \(2\times2\) matrices. The dimensions satisfy
\begin{equation}
  15+21+105+120+189
  +2(1+7+27+35+56)=702.
\end{equation}

For \(H_4/A_3\), the positive side contains \(270\) vertices. The subgroup
preserving the chosen wall and its positive side is
\[
  W(H_3)\cong\operatorname{Alt}(5)\times C_2,
\]
where \(\operatorname{Alt}(5)\) is the alternating group on five letters and
\(C_2\) is the cyclic group of order two. The representation decomposes into
nine components whose irreducible dimensions are
\[
  1,\ 3,\ 3,\ 3,\ 3,\ 4,\ 4,\ 5,\ 5
\]
and whose corresponding multiplicities are
\[
  7,\ 3,\ 3,\ 10,\ 10,\ 10,\ 10,\ 14,\ 7.
\]
On a component of irreducible dimension \(d\) and multiplicity \(m\), the
matrix has the form \(B\otimes I_d\), where \(B\) is an
\(m\times m\) matrix. Consequently, the \(270\)-dimensional positivity test
reduces to nine matrices of orders $7,\ 3,\ 3,\ 10,\ 10,\ 10,\ 10,\ 14,\ 7$.
Indeed,
\begin{equation}
  7+3(3)+3(3)+10(3)+10(3)+10(4)+10(4)+14(5)+7(5)=270.
\end{equation}
The largest matrix that must be tested is
\(14\times14\).

All resulting positivity statements are proved by exact arithmetic over
$\mathbb Q$ or $\mathbb Q(\sqrt5)$. Together with the product construction for reducible Coxeter systems,
these constructions establish the ``if'' direction for every finite Coxeter
system. The edge-reflecting characterization supplies the converse.

Section~\ref{sec:edge-convexity} reviews the graph formulation and
edge-convexity criterion of Ref.~\cite{GaddeJainKulkarni2024}.
Section~\ref{sec:necessity} derives the necessary Coxeter structure from
Marc's mirror-graph theorem. Sections~\ref{sec:parabolic-induction}
and~\ref{sec:classical-certificates} develop parabolic induction and the
classical quotient constructions. Section~\ref{sec:exceptional-certificates}
completes the classification with exact constructions for all six
exceptional systems.

\section{From multi-invariants to edge-convex graphs}
\label{sec:edge-convexity}

A reflecting cut expresses the corresponding invariant as a squared norm.
Edge-convexity combines these
individual reflection-positive decompositions into the positivity condition
needed for monotonicity under local operations.

\subsection{The graph dictionary}
\label{subsec:graph-dictionary}

Let
\begin{equation}
  \mathcal H=\bigotimes_{A=1}^{q}\mathcal H_A
\end{equation}
be the Hilbert space of a system with $q$ parties, where $\mathcal H_A$ is the
Hilbert space of party $A$.  After choosing a basis
$\{\lvert i_A\rangle\}$ of each factor, a pure state takes the form
\begin{equation}
  \lvert\psi\rangle
  =\sum_{i_1,\ldots,i_q}
  \psi_{i_1\cdots i_q}
  \lvert i_1\rangle\otimes\cdots\otimes\lvert i_q\rangle .
  \label{eq:state-components}
\end{equation}
Fix a positive integer $n_r$, let $S_{n_r}$ be the permutation group of
$n_r$ objects, and choose a permutation operator $\sigma_A\in S_{n_r}$
acting on the $n_r$ copies of party $A$.  The associated multi-invariant is
\begin{equation}
  \mathcal Z(\lvert\psi\rangle)
  :=\bigl\langle\psi\bigr\rvert^{\otimes n_r}
  \left(\bigotimes_{A=1}^{q}\sigma_A\right)
  \bigl\lvert\psi\bigr\rangle^{\otimes n_r}.
  \label{eq:multi-invariant-permutations}
\end{equation}
The permutation description of polynomial local-unitary invariants was
developed in Ref.~\cite{GrasslRoettelerBeth1998}. Each index in
Eq.~\eqref{eq:multi-invariant-permutations} is contracted with
an index belonging to the same party, so the expression is invariant under
independent unitary transformations on the $q$ factors of $\mathcal H$.

The contractions are recorded by a finite bipartite edge-labelled
multigraph, where parallel edges with different labels are allowed.  A white
vertex represents a copy of $\psi$, a black vertex a copy of its complex
conjugate $\bar\psi$, and an edge labelled $A$
represents a contraction of two indices belonging to party $A$.  Each vertex
is incident to exactly one edge of every label.  We call any graph with these
properties a \emph{$\psi$-graph}.  The permutation $\sigma_A$ connects the
$r$th white vertex to the $(\sigma_A\cdot r)$th black vertex by an $A$-edge.  A
relabelling of the copies changes the permutations but not the resulting
contraction.  We therefore use the calligraphic letter $\mathcal Z$ for both
the multi-invariant and its $\psi$-graph; the graph has
$n_{\mathcal Z}=n_r$ white vertices.

We study connected graphs throughout.\footnote{The extension to disconnected
graphs stated in Sec.~III.3.3 of Ref.~\cite{GaddeJainKulkarni2024} is
incorrect. For example, the marginal purities
$\mathcal Z_1=\operatorname{Tr}\rho_B^2$ and
$\mathcal Z_2=\operatorname{Tr}\rho_C^2$ of a tripartite pure state have
connected edge-convex graphs. Take both marginal spaces to be six-dimensional.
Set $P=\lvert0\rangle\langle0\rvert$,
$\tau=I_6/6$, and
$\rho_{BC}(t)=tP\otimes\tau+(1-t)\tau\otimes P$ for $0\leq t\leq1$.
Then $F(t):=\widehat{\mathcal Z}_1\widehat{\mathcal Z}_2
=\sqrt{(1+5t^2)(1+5(1-t)^2)}/6$ satisfies $F''(1/2)=-5/27$.
A local measurement on the purifying party realizes a decomposition into
states with nearby values of $t$ on either side of $t=1/2$, violating
monotonicity of $1-F$.
The connected edge-convexity criterion used here is unaffected.}
The operator in parentheses in Eq.~\eqref{eq:multi-invariant-permutations} is
unitary.  For a normalized state, the Cauchy--Schwarz inequality therefore
gives
\begin{equation}
  \bigl\lvert\mathcal Z(\lvert\psi\rangle)\bigr\rvert\leq 1.
  \label{eq:multi-invariant-bound}
\end{equation}
For a fully factorized normalized state every permutation acts on identical
copies, and hence $\mathcal Z(\lvert\psi\rangle)=1$.  
We seek the graphs for which \(\mathcal Z\) is
nonnegative and \(\widehat{\mathcal Z}\) is convex as a function of
\[
  \rho_{\bar A}:=\operatorname{Tr}_A
  \bigl(\lvert\psi\rangle\langle\psi\rvert\bigr)
\]
for every party \(A\). Its complement \(1-\widehat{\mathcal Z}\) is then
concave, and the concavity criterion for pure-state entanglement monotones
implies the LOCC monotonicity condition in Eq.~\eqref{eq:psem}
\cite{Vidal2000,GaddeJainKulkarni2024}.

\subsection{Reflecting cuts and positivity}
\label{subsec:reflecting-cuts}

An automorphism of a $\psi$-graph will always mean a graph automorphism that
preserves every edge label.  It is called \emph{even} if it preserves the
white and black vertex sets and \emph{odd} if it exchanges them.

\begin{lemma}
  \label{lem:fixed-vertex}
  An automorphism of a connected $\psi$-graph that fixes one vertex is the
  identity.
\end{lemma}

\begin{proof}
  Every vertex has a unique neighbour along an edge of a given label.  An
  automorphism fixing a vertex must therefore fix all its neighbours.  The
  same argument applies along any path. Since the graph is connected, the
  automorphism fixes every vertex.
\end{proof}

\begin{definition}[Reflecting cut]
  \label{def:reflecting-cut}
  A nonempty edge set $E$ in a connected $\psi$-graph $\mathcal Z$ is a
  \emph{reflecting cut} if there is an odd automorphism $k$ such that
  \begin{enumerate}
    \item $k$ exchanges the endpoints of every edge in $E$; and
    \item deleting $E$ leaves exactly two connected components, denoted
      $T_1$ and $T_2$, which are exchanged by $k$.
  \end{enumerate}
\end{definition}

The automorphism associated with a reflecting cut is unique and is
automatically a reflection. Indeed, its square fixes the endpoints of every
edge in $E$, so Lemma~\ref{lem:fixed-vertex} implies $k^2=1$. If two
automorphisms $k$ and $k'$ exchanged the endpoints of the edges in $E$, then
$k^{-1}k'$ would fix such an endpoint and hence would be the identity. We
therefore index a reflecting cut by its reflection $k$ when writing the
matrices below.

For example, let
$\lvert\psi\rangle\in\mathcal H_A\otimes\mathcal H_B$ and define the reduced
density matrix of party $A$ by
\begin{equation}
  \rho_A:=\operatorname{Tr}_{\mathcal H_B}
  \bigl(\lvert\psi\rangle\langle\psi\rvert\bigr),
\end{equation}
where $\operatorname{Tr}_{\mathcal H_B}$ is the partial trace over party
$B$.  Denote the four-cycle in Fig.~\ref{fig:square-reflection} by
$\mathcal C_2$.  It represents
$\mathcal C_2(\lvert\psi\rangle)=\operatorname{Tr}(\rho_A^2)$.  Its
horizontal reflection reverses the two $A$-edges.  Removing those edges
leaves two components that are exchanged together with their vertex colours.

\begin{lemma}[Reflection positivity]
  \label{lem:reflection-positivity}
  If the $\psi$-graph $\mathcal Z$ admits a reflecting cut, then
  $\mathcal Z(\lvert\psi\rangle)$ is real and nonnegative for every state
  $\lvert\psi\rangle$.
\end{lemma}

\begin{proof}
  Cut the edges in $E$ and retain their open indices.  The components $T_1$
  and $T_2$ define conjugate open states.  Denoting the first state by
  $\lvert T_1\rangle$, gluing the two halves back together gives
  \begin{equation}
    \mathcal Z(\lvert\psi\rangle)
    =\langle T_1\vert T_1\rangle\geq0.
    \label{eq:reflection-positivity}
  \end{equation}
\end{proof}

For the square graph, the open tensor in this proof is $\rho_A$, and
Eq.~\eqref{eq:reflection-positivity} is the squared-norm identity
$\operatorname{Tr}(\rho_A^2)=\lVert\rho_A\rVert_2^2$.

\begin{figure}[ht]
  \centering
  \begin{tikzpicture}[
      tensor/.style={circle,draw,minimum size=9mm,inner sep=0pt},
      conjugate/.style={tensor,fill=black,text=white},
      every edge/.style={draw,thick}]
    \node[tensor]    (w1) at (-1.5, 0.85) {$\psi_1$};
    \node[conjugate] (b1) at ( 1.5, 0.85) {$\bar\psi_1$};
    \node[conjugate] (b2) at (-1.5,-0.85) {$\bar\psi_2$};
    \node[tensor]    (w2) at ( 1.5,-0.85) {$\psi_2$};
    \draw[thick] (w1)--node[above] {$B$} (b1);
    \draw[thick] (b2)--node[below] {$B$} (w2);
    \draw[thick] (w1)--node[left,pos=0.15]  {$A$} (b2);
    \draw[thick] (b1)--node[right,pos=0.15] {$A$} (w2);
    \draw[dashed] (-2.15,0)--(2.15,0)
      node[right,xshift=4pt,yshift=-7pt] {$k$};
  \end{tikzpicture}
  \caption{The $\psi$-graph for $\operatorname{Tr}(\rho_A^2)$. The dashed
  line is the mirror of $k$ and crosses the reflecting cut, which consists of
  the two $A$-edges. The reflection $k$ exchanges the upper and lower
  components as well as white and black vertices.}
  \label{fig:square-reflection}
\end{figure}

\subsection{The edge-convexity criterion}
\label{subsec:edge-convexity-criterion}

One reflecting cut supplies one positive quadratic form. To prove
monotonicity by this construction, we seek a compatible collection of such
forms that accounts for every pair of equally labelled edges. Throughout,
a pair of edges means
two distinct edges.

Fix a party $A$. For a cut indexed by its reflection $k$, write $R_k$ and
$L_k$ for the edge sets lying on its two sides. The cut \emph{separates} two \(A\)-edges \(e\) and \(e'\) if they lie on
opposite sides of the cut; we orient the two sides so that
\(e\in R_k\) and \(e'\in L_k\).  To
such a separated pair assign a coefficient $\mathcal M^{(k)}_{e,e'}$.

\begin{definition}[Edge-convexity]
  \label{def:edge-convexity}
  A connected $\psi$-graph is \emph{$A$-edge-convex} if the coefficients
  $\mathcal M^{(k)}_{e,e'}$ can be chosen so that
  \begin{equation}
    \sum_{\substack{k:\,e\in R_k,\ e'\in L_k}}
    \mathcal M^{(k)}_{e,e'}=1
    \qquad
    \text{for every ordered pair of $A$-edges }e\neq e',
    \label{eq:edge-convexity}
  \end{equation}
  where the sum is over all reflecting cuts that separate $e$ and $e'$. For
  every $k$ in the sum, the following matrix must be real, symmetric, and
  positive semidefinite:
  \begin{equation}
    \mathcal P^{(k)}(e,f):=\mathcal M^{(k)}_{e,k(f)},
    \qquad e,f\in R_k,
    \qquad \mathcal P^{(k)}\succeq0
    \label{eq:cut-psd}
  \end{equation}
  Here $e$ and $f$ are $A$-edges and
  $\mathcal P^{(k)}\succeq0$ means that
  $\boldsymbol v^{\mathsf T}\mathcal P^{(k)}\boldsymbol v\geq0$ for every
  real vector $\boldsymbol v$.  The graph is \emph{edge-convex} if it is
  $A$-edge-convex for every party $A$.
\end{definition}

Equation~\eqref{eq:edge-convexity} is a partition of unity over cuts, while
Eq.~\eqref{eq:cut-psd} couples the weights assigned to different edge pairs.

In Fig.~\ref{fig:square-reflection}, the horizontal cut separates the two
$B$-edges.  The analogous vertical cut separates the two $A$-edges.  Assigning
the coefficient $\mathcal M^{(k)}_{e,e'}=1$ to each separated pair gives the
$1\times1$ matrix $\mathcal P^{(k)}=(1)$ and satisfies
Eqs.~\eqref{eq:cut-psd} and \eqref{eq:edge-convexity}.  The square is therefore
edge-convex.

This condition implies monotonicity by the following result of
Ref.~\cite{GaddeJainKulkarni2024}.

\begin{theorem}[Edge-convexity criterion]
  \label{thm:edge-convexity-criterion}
  Let $\mathcal Z$ be a connected edge-convex $\psi$-graph.  For normalized
  pure states,
  \begin{equation}
    \widehat\nu(\mathcal Z):=1-\widehat{\mathcal Z}
    \label{eq:edge-convex-monotone}
  \end{equation}
  is a pure-state entanglement monotone.
\end{theorem}

The proof reorganizes the change under a local quantum operation into terms
labelled by reflecting cuts.  The partition of unity in
Eq.~\eqref{eq:edge-convexity} reconstructs the full change, while reflection
positivity and Eq.~\eqref{eq:cut-psd} give each term the sign required by
Eq.~\eqref{eq:psem}.

The normalization equation also requires the following separation property.

\begin{definition}[Edge-reflecting graph]
  \label{def:edge-reflecting}
  A connected $\psi$-graph is \emph{$A$-edge-reflecting} if every pair of
  $A$-edges is separated by a reflecting cut.  It is \emph{edge-reflecting}
  if this holds for every party $A$.
\end{definition}

\begin{lemma}
  \label{lem:convex-implies-reflecting}
  Every edge-convex $\psi$-graph is edge-reflecting.
\end{lemma}

\begin{proof}
  If no reflecting cut separated a pair $e,e'$ of $A$-edges, the sum on the
  left-hand side of Eq.~\eqref{eq:edge-convexity} would be empty and hence
  zero, contradicting its required value $1$.
\end{proof}

The next section classifies connected graphs with this edge-reflecting
property, including the possible duplication of colour families.

\section{Necessity: mirror graphs and Coxeter systems}
\label{sec:necessity}

Section~\ref{sec:edge-convexity} reduced the necessary part of the
classification to a graph-theoretic question: which connected $\psi$-graphs
are edge-reflecting? The separation property forces every edge to belong to
a reflecting cut. The characterization of mirror graphs then gives the
Coxeter structure.

\subsection{From edge-reflecting to mirror graphs}
\label{subsec:edge-reflecting-mirror}

We use the notion of a mirror graph from
Ref.~\cite{BresarKlavzarLipovecMohar2004}, with the additional requirement
that the reflecting automorphisms preserve every edge label.
\begin{definition}[Mirror graph]
  \label{def:mirror-graph}
  A connected \(\psi\)-graph is a \emph{mirror graph} if its edge set admits
  a partition into reflecting cuts; equivalently, every edge belongs to
  exactly one reflecting cut.
\end{definition}

\begin{lemma}[One-crossing property]
  \label{lem:one-crossing}
  A geodesic contains at most one edge of any given reflecting cut.  Here a
  geodesic is a path of minimum length between its endpoints.
\end{lemma}

\begin{proof}
  Suppose that a geodesic crosses the same reflecting cut twice, and consider
  its first two crossing edges.  The reflection exchanges the endpoints of
  both edges.  Reflecting the portion of the path between the crossings
  therefore produces a path with the same endpoints but two fewer edges.
  This contradicts the defining minimality of a geodesic.
\end{proof}

We can now identify the edge-reflecting condition with the mirror property.

\begin{proposition}
  \label{prop:edge-reflecting-mirror}
  A connected $\psi$-graph is edge-reflecting if and only if it is a mirror
  graph.
\end{proposition}

\begin{proof}
  Suppose first that $\mathcal Z$ is edge-reflecting, and take two distinct
  vertices $u$ and $v$.  If their incident $A$-edges are distinct for some
  party $A$, a reflecting cut separating those edges also separates $u$ from
  $v$.  If their incident edges coincide for every party, then $u$ and $v$
  form a two-vertex connected component.  Connectedness implies that these
  are the only vertices, and the set of all edges is a reflecting cut
  separating them.  Thus every pair of distinct vertices can be separated by
  a reflecting cut.  In particular, a cut separating the endpoints of an
  arbitrary edge must contain that edge.  Since an edge cannot belong to two distinct reflecting cuts, the edge set admits
  a partition into reflecting cuts. Hence $\mathcal Z$ is a mirror
  graph.

  Conversely, suppose that $\mathcal Z$ is a mirror graph.  Let
  $e_1=(u_1,v_1)$ and $e_2=(u_2,v_2)$ be distinct $A$-edges, with $u_1$ and
  $u_2$ chosen to minimize the distance between the two pairs of endpoints.
  Let $p$ be a geodesic from $u_1$ to $u_2$.  Because $\mathcal Z$ is
  bipartite, appending $e_1$ or $e_2$ to the corresponding end of $p$ again
  gives a geodesic: minimality forbids a shorter path, while bipartiteness
  fixes the opposite parity of its length.

  Choose any edge of $p$.  It belongs to a reflecting cut by the mirror
  property.  Lemma~\ref{lem:one-crossing}, applied also to the two extended
  geodesics, shows that this cut contains no other edge of $p$ and contains
  neither $e_1$ nor $e_2$.  The cut therefore places $e_1$ and $e_2$ on
  opposite sides.  Every pair of $A$-edges can be separated in this way, for
  every party $A$, so $\mathcal Z$ is edge-reflecting.
\end{proof}

\subsection{The Coxeter characterization}
\label{subsec:mirror-coxeter}
For the purpose of characterizing mirror graphs, we may assume without loss
of generality that the graph is simple. Here \emph{simple} means that the
graph has neither loops nor parallel edges. Loops are already excluded
because a \(\psi\)-graph is bipartite. Suppose that an \(A\)-edge and a
\(B\)-edge have the same pair of endpoints. Reflections in the cuts along a
path act transitively on the vertices and preserve the edge labels; hence the
\(A\)-edge and \(B\)-edge incident to every vertex have the same neighbouring
endpoint. These globally parallel families can therefore be collapsed to a
single colour in a reduced graph. Every reflecting cut contains either all
or none of the parallel edges between two vertices. Collapsing these edges
leaves the two vertex sets separated by the cut unchanged, and the same
vertex reflection acts on the reduced graph. Conversely, each
label-preserving reflection of the reduced graph extends to the original
graph by mapping restored edges to edges of the same colour. Reflecting
cuts therefore correspond under the reduction.

For each original colour, this correspondence preserves the edge pairs
separated by each cut and the action of its reflection. The same cut
matrices then satisfy Eqs.~\eqref{eq:edge-convexity} and~\eqref{eq:cut-psd}
in both graphs, so collapsing globally parallel colour families preserves
edge-convexity.

In the tensor contraction, merging colours $A$ and $B$ amounts to treating
the two parties as one, with Hilbert space $\mathcal H_A\otimes\mathcal H_B$.
The contraction is unchanged. Monotonicity under joint operations on this
combined party implies monotonicity under separate local operations on
$A$ and $B$.

Simple mirror graphs have a complete group-theoretic characterization.

\begin{theorem}[Mirror-graph characterization]
  \label{thm:mirror-coxeter}
  A simple connected $\psi$-graph $\mathcal Z$ is a mirror graph if and only if
  there is a finite Coxeter system $(W,S)$ such that
  \begin{equation}
    \mathcal Z\cong\operatorname{Cay}(W,S),
    \label{eq:mirror-coxeter}
  \end{equation}
  where the isomorphism preserves the edge labels.
\end{theorem}

This is Theorem~2.8 of Ref.~\cite{Marc2017}, specialized to the labelled
bipartite graphs used here.
Reflections in the mirrors send any chosen vertex to any other by a unique
group element, so the graph is a Cayley graph with involutive generators.  For
two generators $s,t\in S$, the mirror structure selects a unique cycle whose
labels alternate between $s$ and $t$.  The cycle is convex, meaning that every
shortest path between two of its vertices stays within the cycle.  If its
length is $2m_{st}$, reading its edge labels gives the relation
$(st)^{m_{st}}=1$.  Filling these cycles by disks allows every loop to be
reduced to a point, so the pairwise relations generate all relations.  This
is precisely the Coxeter presentation in
Eq.~\eqref{eq:coxeter-presentation}.

To match the labels, choose a base vertex and associate each standard
generator with the label of the incident edge that it reverses.  Graph
reflections preserve every edge label, so their action propagates this
identification to every vertex.  Conversely, in a Coxeter Cayley graph the
edge cuts reversed by the reflections $wsw^{-1}$, with $w\in W$ and
$s\in S$, are reflecting cuts, and every edge belongs to one of them.  Since
$\mathcal Z$ has finitely many vertices, the corresponding Coxeter group $W$
is finite.

Combining the preceding results proves the necessary direction of
Theorem~\ref{thm:classification}.

\begin{proposition}[Necessity]
  \label{prop:coxeter-necessity}
  If a connected $\psi$-graph $\mathcal Z$ is edge-convex, then there is a
  finite Coxeter system $(W,S)$ whose standard Cayley graph is
  label-preservingly isomorphic to the graph obtained from $\mathcal Z$
  by merging globally parallel colour families.
\end{proposition}

\begin{proof}
  By the preceding reduction, it suffices to consider a simple graph
  $\mathcal Z$. Lemma~\ref{lem:convex-implies-reflecting},
  Proposition~\ref{prop:edge-reflecting-mirror}, and
  Theorem~\ref{thm:mirror-coxeter} give
  \begin{equation}
    \mathcal Z\ \text{edge-convex}
    \;\Longrightarrow\;
    \mathcal Z\ \text{edge-reflecting}
    \;\Longleftrightarrow\;
    \mathcal Z\ \text{a mirror graph}
    \;\Longleftrightarrow\;
    \mathcal Z\cong\operatorname{Cay}(W,S).
    \label{eq:necessity-chain}
    \qedhere
  \end{equation}
\end{proof}

The Coxeter structure is therefore necessary, but it does not yet establish
edge-convexity.  The remaining direction requires compatible
positive-semidefinite matrices satisfying the normalization conditions for
all reflecting cuts of the Coxeter graph. We construct them using parabolic
quotients and induction.

\section{Vertex-convexity of quotients and parabolic induction}
\label{sec:parabolic-induction}

To prove the sufficient direction of Theorem~\ref{thm:classification}, we need
matrices on the reflecting cuts of $\operatorname{Cay}(W,S)$ described in
Section~\ref{sec:necessity}. Their size grows with the order of $W$. We
reduce their size by grouping the graph into parabolic fibres. A wall then
cuts a fibre internally or separates two different fibres. We treat the first
case recursively and the second with a matrix on the
parabolic quotient.

\subsection{Vertex-convexity on a parabolic quotient}
\label{subsec:vertex-convexity}

Fix a subset $I\subseteq S$. The group $W_I=\langle I\rangle$ is the
standard parabolic subgroup generated by $I$, and
$X_I=W/W_I$ is its set of left cosets.  Retaining only the edges labelled by
$I$ decomposes $\operatorname{Cay}(W,S)$ into copies of
$\operatorname{Cay}(W_I,I)$ indexed by $X_I$.

To describe how a reflecting wall acts on these copies, use the standard real
reflection representation of $W$ on a Euclidean vector space $V$.  Let
$\Phi\subset V$ be a root system for $W$: every $\alpha\in\Phi$ defines the
reflection $s_\alpha$ in the hyperplane perpendicular to $\alpha$, and every
reflection in $W$ arises in this way. To realize $X_I$ as a set of vectors,
let $\alpha_s$ be the simple root associated with $s\in S$, and choose
$\lambda_I$ in the relative interior of the corresponding face of the
fundamental chamber:
\[
  (\lambda_I,\alpha_s)=0\quad(s\in I),\qquad
  (\lambda_I,\alpha_s)>0\quad(s\in S\setminus I).
\]
Its stabilizer, the subgroup that fixes $\lambda_I$, is precisely $W_I$
\cite{Humphreys1990}. Identify the coset $wW_I$ with
$w\lambda_I$. We can then regard the elements of $X_I$ as vectors with
the $W$-invariant inner product
$(\,\cdot\,,\,\cdot\,)$. The positive side of the wall perpendicular to
$\alpha$ is
\begin{equation}
  X_I^+(\alpha):=\{x\in X_I:(x,\alpha)>0\}.
  \label{eq:positive-quotient-side}
\end{equation}
Vertices with $(x,\alpha)=0$ lie in the wall.

This discussion is a bit abstract so let us explain it with a simple geometric example. 
\paragraph{The quotient \(A_2/A_1\).}
Let \(W=S_3\), the Coxeter group of type \(A_2\), act on
\[
  V:=\{(x_1,x_2,x_3)\in\mathbb R^3:x_1+x_2+x_3=0\}
\]
by permuting the coordinates. Take
\[
  S=\{s_1,s_2\},\qquad s_1=(12),\qquad s_2=(23),
\]
and choose \(I=\{s_1\}\). Then \(W_I=\langle s_1\rangle\) has type \(A_1\),
and retaining only the \(s_1\)-edges decomposes the six-vertex Cayley graph
of \(S_3\) into three two-vertex fibres.
The vector
\[
  \lambda_I=(1,1,-2)
\]
is fixed precisely by \(W_I\). Its orbit is
\[
  X_I\cong W\lambda_I
  =
  \{(1,1,-2),\ (1,-2,1),\ (-2,1,1)\},
\]
which consists of the vertices of an equilateral triangle in \(V\).

Let \(e_1,e_2,e_3\) be the standard basis vectors of \(\mathbb R^3\).
The root \(\alpha=e_1-e_2\) defines the reflection exchanging the first
two coordinates. Its wall is the line
\[
  \alpha^\perp=\{x\in V:x_1=x_2\}.
\]
The vertex \((1,1,-2)\) lies in the wall, while the other two vertices lie
on opposite sides:
\[
  \bigl((1,-2,1),\alpha\bigr)=3,
  \qquad
  \bigl((-2,1,1),\alpha\bigr)=-3.
\]
Hence
\[
  X_I^+(\alpha)=\{(1,-2,1)\}.
\]
Geometrically, the reflection fixes the vertex lying in the wall and
exchanges the two vertices on its opposite sides.

We adapt the vertex-convexity condition of
Ref.~\cite{GaddeJainKulkarni2024}, Definition~IV.1, to parabolic quotients.

\begin{definition}[Vertex-convexity]
  \label{def:vertex-convexity}
  The quotient $X_I$ is \emph{vertex-convex} if each root $\alpha$ has a real
  symmetric matrix $Q^{(\alpha)}$ indexed by $X_I^+(\alpha)$ such that
  \begin{equation}
    Q^{(\alpha)}\succeq0
    \label{eq:quotient-psd}
  \end{equation}
  and, for every ordered pair of distinct vertices $x,y\in X_I$,
  \begin{equation}
    \sum_{\alpha:\,(x,\alpha)>0>(y,\alpha)}
    Q^{(\alpha)}_{x,s_\alpha y}=1.
    \label{eq:vertex-convexity}
  \end{equation}
\end{definition}
\noindent 
Without loss of generality the matrices can be taken to be $W$-equivariant:
  \begin{equation}
    Q^{(w\alpha)}_{wx,wy}=Q^{(\alpha)}_{x,y}
    \qquad (w\in W).
    \label{eq:quotient-equivariance}
  \end{equation}
Indeed, any family
of matrices satisfying these conditions can be symmetrized by defining
\begin{equation}
  \overline Q^{(\alpha)}_{x,y}
  :=
  \frac{1}{|W|}
  \sum_{w\in W}
  Q^{(w\alpha)}_{wx,wy}.
\end{equation}
Each term is a permutation-conjugate of a positive-semidefinite matrix, so
\(\overline Q^{(\alpha)}\) is positive semidefinite. Since the normalization
conditions are linear and \(W\)-invariant, they are also preserved by this
average. The averaged family satisfies
\begin{equation}
  \overline Q^{(w\alpha)}_{wx,wy}
  =
  \overline Q^{(\alpha)}_{x,y}.
\end{equation}
We therefore impose \(W\)-equivariance, so it is enough
to construct one matrix for each \(W\)-orbit of roots.

Both $\alpha$ and $-\alpha$ belong to $\Phi$, but the inequalities in
Eq.~\eqref{eq:vertex-convexity} select one orientation of each separating
wall.  The equation is the quotient counterpart of
Eq.~\eqref{eq:edge-convexity}: quotient vertices replace equally labelled
edges.
\paragraph{The quotient \(A_2/A_1\).}
The quotient $A_2/A_1$ is a triangle.  Each reflection swaps two vertices and fixes the
third, so its positive side contains one vertex and the matrix is
$Q^{(\alpha)}=(1)$.  Every ordered pair is separated by exactly one wall;
hence Eq.~\eqref{eq:vertex-convexity} holds.  The parabolic fibre has type
$A_1$ and contains one $s_1$-edge, so it is edge-convex because there is no
distinct pair of $s_1$-edges.  The induction theorem below lifts this
one-dimensional quotient matrix to the $s_1$-edges of the six-cycle
$\operatorname{Cay}(A_2,S)$; the other endpoint parabolic treats $s_2$.

\subsection{Symmetry reduction of quotient matrices}
\label{subsec:invariant-class-certificates}

We use the standard symmetry reduction of invariant semidefinite problems
\cite{GatermannParrilo2004}. The wall-preserving symmetry reduces both conditions in
Definition~\ref{def:vertex-convexity} to finite exact calculations.  Let $a$
label a conjugacy class of reflections, that is, a class under $t\mapsto wtw^{-1}$.
Choose a root $\alpha_a$ whose reflection belongs to that class and define
its oriented wall stabilizer by
\begin{equation}
  H_a:=\{w\in W:w\alpha_a=\alpha_a\}.
  \label{eq:oriented-wall-stabilizer}
\end{equation}
This group acts diagonally by $w(x,z)=(wx,wz)$ on
$X_I^+(\alpha_a)\times X_I^+(\alpha_a)$.  Partition that product into
symmetric classes $C_{a,0},\ldots,C_{a,r_a}$ of vertex pairs, for some $r_a\geq0$.
Each class must be preserved by this action. Here symmetric means that $(x,z)\in C_{a,i}$ implies
$(z,x)\in C_{a,i}$.  Let $A_{a,i}$ be the $0/1$ matrix indicating membership
in $C_{a,i}$ and seek real coefficients $c_{a,i}$ such that
\begin{equation}
  Q_a=\sum_{i=0}^{r_a}c_{a,i}A_{a,i}.
  \label{eq:invariant-class-matrix}
\end{equation}

For the quotient \(A_2/A_1\), each positive wall side contains only one
vertex, so the symmetry reduction is trivial. The quotient \(B_3/A_2\)
has larger wall matrices and illustrates the reduction.
\paragraph{The quotient \(B_3/A_2\).}
The quotient \(X=W(B_3)/W(A_2)\) can be realized as the cube
\begin{equation}
  X=\{(\epsilon_1,\epsilon_2,\epsilon_3):
  \epsilon_i\in\{\pm1\}\}.
  \label{eq:b3-a2-cube}
\end{equation}
Choose the short root \(\alpha=e_1\). Its positive side is the
four-vertex square
\begin{equation}
  X^+(\alpha)
  =\{(1,\epsilon_2,\epsilon_3):
  \epsilon_2,\epsilon_3\in\{\pm1\}\}.
  \label{eq:b3-a2-positive-side}
\end{equation}
The oriented wall stabilizer is \(H_\alpha\cong W(B_2)\); it acts on this
square by signed permutations of the last two coordinates.

For \(x,z\in X^+(\alpha)\), let \(d_H(x,z)\) be their Hamming distance,
the number of coordinates in which they differ. The diagonal action of \(H_\alpha\) has three
symmetric classes,
\begin{equation}
  C_i:=\{(x,z):d_H(x,z)=i\},
  \qquad i=0,1,2.
  \label{eq:b3-a2-classes}
\end{equation}
In the order
\[
  (1,1,1),\quad(1,1,-1),\quad
  (1,-1,1),\quad(1,-1,-1),
\]
their indicator matrices are
\begin{equation}
  A_0=I_4,\qquad
  A_1=
  \begin{pmatrix}
    0&1&1&0\\
    1&0&0&1\\
    1&0&0&1\\
    0&1&1&0
  \end{pmatrix},
  \qquad
  A_2=
  \begin{pmatrix}
    0&0&0&1\\
    0&0&1&0\\
    0&1&0&0\\
    1&0&0&0
  \end{pmatrix}.
  \label{eq:b3-a2-class-matrices}
\end{equation}
Thus the wall matrix has the three-parameter form
\begin{equation}
  Q=c_0A_0+c_1A_1+c_2A_2.
  \label{eq:b3-a2-invariant-matrix}
\end{equation}

The partition-of-unity condition becomes a linear system.  Let $\tau$ label a $W$-orbit
of ordered pairs of distinct quotient vertices and choose a representative
$(x_\tau,y_\tau)$.  For an oriented root $\beta$ of type $a$ that separates
this pair, choose $g\in W$ with $g\alpha_a=\beta$.  Define
$N_{\tau,a,i}$ to be the number of these roots for which
\begin{equation}
  \bigl(g^{-1}x_\tau,
  g^{-1}s_\beta y_\tau\bigr)\in C_{a,i}.
  \label{eq:invariant-class-count}
\end{equation}
The number is independent of the choice of $g$ because two such choices
differ by an element of $H_a$.

Collect the coefficients $c_{a,i}$ into a vector $c$ and the counts
$N_{\tau,a,i}$ into a matrix $N$. Let $\boldsymbol 1$ denote the column
vector whose entries are all one.

\begin{lemma}[Symmetry-reduced vertex-convexity]
  \label{lem:invariant-class-certificate}
  Suppose the coefficients in Eq.~\eqref{eq:invariant-class-matrix} satisfy
  \begin{equation}
    Nc=\boldsymbol 1
    \label{eq:class-normalization}
  \end{equation}
  and $Q_a\succeq0$ for every reflection type $a$.  Then $X_I$ is
  vertex-convex.
\end{lemma}

\begin{proof}
  If $\beta=g\alpha_a$, transport $Q_a$ to the wall perpendicular to $\beta$
  by
  \begin{equation}
    Q^{(\beta)}_{x,z}:=(Q_a)_{g^{-1}x,g^{-1}z}.
    \label{eq:transport-quotient-matrix}
  \end{equation}
  Invariance under $H_a$ makes this definition independent of $g$ and gives
  Eq.~\eqref{eq:quotient-equivariance}.  For the representative pair
  $(x_\tau,y_\tau)$, the left-hand side of
  Eq.~\eqref{eq:vertex-convexity} is
  $\sum_{a,i}N_{\tau,a,i}c_{a,i}$, which equals one by
  Eq.~\eqref{eq:class-normalization}.  Equivariance gives the same identity
  for every pair.  Finally, transport preserves positive semidefiniteness, so
  Eq.~\eqref{eq:quotient-psd} follows from $Q_a\succeq0$.
\end{proof}

For a fixed pair of distinct quotient vertices,
\(N_{\tau,a,i}\) counts the separating walls of type \(a\) for which the
following procedure gives a pair in \(C_{a,i}\): reflect the second vertex
to the positive side, then transport both vertices to the representative
wall. Thus each row of
\(N\) records how many times every coefficient \(c_{a,i}\) appears in the
normalization condition for one type of vertex pair. The equation
\(Nc=\boldsymbol 1\) states that these contributions sum to one.
\paragraph{The quotient \(B_3/A_2\).}
For the cube, the \(W(B_3)\)-orbits of distinct vertex pairs are labelled
by their Hamming distance \(k=1,2,3\). A pair at distance \(k\) is separated
by exactly \(k\) short-root walls. Reflecting the second vertex in any one
of these walls corrects one differing coordinate, so the resulting
positive-side pair belongs to \(C_{k-1}\). We take the matrices associated
with long-root walls to be zero. The nonzero part of the normalization
matrix is therefore
\begin{equation}
  N=
  \begin{pmatrix}
    1&0&0\\
    0&2&0\\
    0&0&3
  \end{pmatrix}.
  \label{eq:b3-a2-normalization-matrix}
\end{equation}
The equation \(Nc=\boldsymbol1\) uniquely gives
\begin{equation}
  c_0=1,\qquad c_1=\frac12,\qquad c_2=\frac13.
  \label{eq:b3-a2-coefficients}
\end{equation}
Thus, within the ansatz \(Q^{(\alpha)}=0\) for every long root \(\alpha\),
the normalization equations uniquely determine the short-root wall matrix
\begin{equation}
  Q=A_0+\frac12A_1+\frac13A_2.
  \label{eq:b3-a2-certificate}
\end{equation}
Positive semidefiniteness can be checked directly in this four-dimensional
example. For larger matrices, block diagonalization under the wall
stabilizer reduces the size of the positivity test.

The classes $C_{a,i}$ need not be individual orbits of $H_a$ on ordered
pairs, which are conventionally called \emph{orbitals}.  They may be
symmetric invariant unions of orbitals.  This freedom is useful when inner
products give a coarser description than the full orbital decomposition.

The invariant-class expansion reduces the number of independent entries of
\(Q_a\), but it does not block-diagonalize the matrix. A second,
representation-theoretic reduction simplifies the positivity test.
The action of \(H_a\) on \(X_I^+(\alpha_a)\) induces a permutation
representation on \(\mathbb R[X_I^+(\alpha_a)]\), the vector space with one
basis vector for each positive-side vertex. The group permutes these basis
vectors in the same way as the vertices. Since \(Q_a\) is \(H_a\)-invariant, it
commutes with this group action. Decompose the permutation representation
into irreducible representations. If an irreducible representation
\(V_\rho\) of dimension \(d_\rho\) occurs \(m_\rho\) times, the sum of these
copies has the form
\[
  \mathbb R^{m_\rho}\otimes V_\rho.
\]
For the real representations occurring below, which remain irreducible
after complexification,
Schur's lemma says that on this space \(Q_a\) acts as
\[
  B_\rho\otimes I_{d_\rho},
\]
where \(B_\rho\) is an \(m_\rho\times m_\rho\) real symmetric matrix. We call
\(B_\rho\) the \emph{multiplicity block}. Thus a matrix acting on a space of
dimension \(m_\rho d_\rho\) is reduced to a positivity test of order
\(m_\rho\). In the cases considered below,
\[
  Q_a\succeq0
  \quad\Longleftrightarrow\quad
  B_\rho\succeq0
  \quad\text{for every }\rho.
\]
\paragraph{The quotient \(B_3/A_2\).}
Define the orthonormal vectors
\begin{align}
  u_0&=\frac12(1,1,1,1),&
  u_2&=\frac12(1,1,-1,-1),\nonumber\\
  u_3&=\frac12(1,-1,1,-1),&
  u_{23}&=\frac12(1,-1,-1,1).
  \label{eq:b3-a2-walsh-basis}
\end{align}
The constant vector \(u_0\), the two-dimensional space
\(\operatorname{span}\{u_2,u_3\}\), and the line spanned by \(u_{23}\)
are the three irreducible components of the wall-stabilizer permutation
representation. On these components, \(A_1\) has eigenvalues \(2,0,-2\),
while \(A_2\) has eigenvalues \(1,-1,1\), respectively. Therefore
\begin{equation}
  Q
  \cong
  \left(\frac73\right)
  \oplus
  \left(\frac23 I_2\right)
  \oplus
  \left(\frac13\right).
  \label{eq:b3-a2-block-diagonalization}
\end{equation}
All three blocks are positive definite. The factor \(I_{d_\rho}\) accounts
for the wall matrix acting as the scalar \(2/3\) on the
two-dimensional irreducible component.

\subsection{Parabolic induction}
\label{subsec:parabolic-induction-theorem}

We now lift the matrices establishing vertex-convexity of a quotient to the
full Cayley graph.

\begin{theorem}[Parabolic induction]
  \label{thm:parabolic-induction}
  Let $I\subseteq S$.  If $\operatorname{Cay}(W_I,I)$ is edge-convex and
  $X_I=W/W_I$ is vertex-convex, then $\operatorname{Cay}(W,S)$ is
  edge-convex for every label $s\in I$.
\end{theorem}

\begin{proof}
  Retaining only the $I$-labelled edges decomposes
  $\operatorname{Cay}(W,S)$ into parabolic fibres $F_x$ indexed by
  $x\in X_I$, with
  \begin{equation}
    F_x\cong\operatorname{Cay}(W_I,I).
    \label{eq:parabolic-fibre}
  \end{equation}
  Identify the graph vertices with chambers. The chambers in $F_x$ are
  precisely those incident to the face represented by $x=w\lambda_I$.
  If $(x,\alpha)\neq0$, they all lie on the same side of $\alpha^\perp$:
  each has $x$ in its closure, and a chamber cannot cross a reflecting
  hyperplane. If $(x,\alpha)=0$, the wall contains that face and cuts its
  incident chambers internally. This argument uses the choice of
  $\lambda_I$ in Section~\ref{subsec:vertex-convexity}.
  The reflection $s_\alpha$
  exchanges a positive-side fibre $F_x$ with the negative-side fibre
  $F_{s_\alpha x}$.  A fibre with $(x,\alpha)=0$ is preserved by the
  reflection.  If $x=wW_I$ is such a vertex, then
  $w^{-1}s_\alpha w\in W_I$ is a reflection of the parabolic subgroup, so the
  wall restricts to a reflecting cut inside $F_x$.

  Fix a label $s\in I$, and let $J$ be the all-ones matrix indexed by the
  $s$-edges in one parabolic fibre.  On the $s$-edges lying in fibres strictly
  on the positive side, use
  \begin{equation}
    Q^{(\alpha)}\otimes J.
    \label{eq:tensor-lift}
  \end{equation}
  On each fibre lying in the wall, use the cut matrix supplied by the
  edge-convexity of $\operatorname{Cay}(W_I,I)$.  The full cut matrix is block
  diagonal, with the tensor-product block and these internal blocks.  It is
  positive semidefinite because $Q^{(\alpha)}$, $J$, and all internal blocks
  are positive semidefinite.

  To check the partition of unity, first take two $s$-edges in the same
  fibre. Their coefficients come from the cut matrices of the parabolic
  graph: every reflecting cut of that fibre extends uniquely to a
  reflecting wall of $W$ through the fibre.
  If the edges lie in different fibres, their quotient vertices are distinct
  and Eq.~\eqref{eq:vertex-convexity} gives the required sum; the corresponding
  entry of $J$ is one, independently of the positions inside the fibres.
  If one quotient vertex lies in the wall, the two edges belong to different
  blocks and their coefficient is zero. The same is true for two distinct
  wall-fixed fibres. Thus only walls strictly separating the two quotient
  vertices contribute to the normalization for different fibres.
  The lifted matrices therefore satisfy
  Eq.~\eqref{eq:edge-convexity} for the label $s$.
\end{proof}

Applying the theorem to parabolics whose labels cover $S$ gives the form used
throughout the classification.

\begin{corollary}[Parabolic cover]
  \label{cor:parabolic-cover}
  Suppose $I_1,\ldots,I_m\subseteq S$ satisfy
  \begin{equation}
    S=I_1\cup\cdots\cup I_m.
    \label{eq:parabolic-cover}
  \end{equation}
  If every $\operatorname{Cay}(W_{I_j},I_j)$ is edge-convex and every quotient
  $W/W_{I_j}$ is vertex-convex, then $\operatorname{Cay}(W,S)$ is
  edge-convex.
\end{corollary}

\begin{proof}
  Theorem~\ref{thm:parabolic-induction} proves edge-convexity for all labels
  in $I_j$.  Equation~\eqref{eq:parabolic-cover} covers every label in $S$.
\end{proof}

\subsection{Products and the irreducible reduction}
\label{subsec:coxeter-products}

The product theorem of Ref.~\cite{GaddeJainKulkarni2024},
Theorem~I.6, also follows from parabolic induction.

\begin{proposition}[Products]
  \label{prop:coxeter-products}
  Let $(W_1,S_1)$ and $(W_2,S_2)$ be finite Coxeter systems.  If
  $\operatorname{Cay}(W_1,S_1)$ and
  $\operatorname{Cay}(W_2,S_2)$ are edge-convex, then so is
  \begin{equation}
    \operatorname{Cay}
    \bigl(W_1\times W_2,S_1\sqcup S_2\bigr).
    \label{eq:coxeter-product}
  \end{equation}
  Here $S_1\sqcup S_2$ denotes the disjoint union of the generating sets.
\end{proposition}

\begin{proof}
  Take $I=S_1$. Then $W_I=W_1\times\{1\}$, and the quotient
  $(W_1\times W_2)/W_I$ identifies with the vertex set of
  $\operatorname{Cay}(W_2,S_2)$. By Lemma~IV.1 of
  Ref.~\cite{GaddeJainKulkarni2024}, edge-convexity of this graph
  implies vertex-convexity. Theorem~\ref{thm:parabolic-induction}
  therefore gives edge-convexity for every label in $S_1$.
  Interchanging the two factors treats every label in $S_2$.
\end{proof}

The Coxeter group of a disconnected diagram is a direct product of the
groups associated with its connected components. Proposition~\ref{prop:coxeter-products}
therefore reduces sufficiency to irreducible finite Coxeter systems.  For an
irreducible system, Corollary~\ref{cor:parabolic-cover} reduces the problem to
a finite collection of vertex-convexity proofs for quotients.  The next section gives
the rank-two and classical constructions, which provide the base cases for
the exceptional groups.

\section{Rank-two and classical quotient constructions}
\label{sec:classical-certificates}

For the infinite families, we construct quotient matrices whose positivity
follows from a single kernel on sign vectors. The resulting proof is
analytic and requires no numerical positivity checks.

\begin{theorem}[Rank-two and classical systems]
  \label{thm:classical-edge-convexity}
  The standard Cayley graphs of the finite Coxeter systems
  \begin{equation}
    I_2(m),\qquad A_n,\qquad B_n,\qquad D_n
    \label{eq:classical-families}
  \end{equation}
  are edge-convex for $m\geq3$, $n\geq1$ in type $A$, $n\geq2$ in type
  $B$, and $n\geq4$ in type $D$.
\end{theorem}

\subsection{Rank two}
\label{subsec:rank-two-certificate}

We recall the even-cycle construction of
Ref.~\cite{GaddeJainKulkarni2024}, Sec.~III.3.2.

The Cayley graph of $I_2(m)$, with simple reflections $s,t$ satisfying
$(st)^m=1$, is a cycle with $2m$ edges, alternately labelled by $s$ and $t$.
Any two edges with the same label are exchanged by a unique reflection of
the cycle. Its mirror
separates the two edges.  For each mirror and each label, take
$\mathcal P^{(k)}$ to be the identity matrix on the edges of that label on
one side of the mirror.  The matrix is positive semidefinite.  Its entry for
a pair of edges is one precisely when the reflection exchanges the pair and
is zero otherwise.  The unique exchanging reflection therefore gives
Eq.~\eqref{eq:edge-convexity}.  Hence $I_2(m)$ is edge-convex.

\subsection{A positive Hamming kernel}
\label{subsec:hamming-kernel}

On the sign cube $\{\pm1\}^N$, define the Hamming distance between $x$ and
$z$ by
\begin{equation}
  d_H(x,z):=\#\{j:x_j\neq z_j\}.
  \label{eq:hamming-distance}
\end{equation}
The following lemma supplies all nontrivial positivity statements in this
section.

\begin{lemma}[Hamming kernel]
  \label{lem:hamming-kernel}
  For every positive integer $r$, the matrix on $\{\pm1\}^N$ with entries
  \begin{equation}
    K_r(x,z):=\binom{d_H(x,z)+r}{r}^{-1}
    \label{eq:hamming-kernel}
  \end{equation}
  is positive semidefinite.  The same is true after restricting the matrix
  to any subset of $\{\pm1\}^N$.
\end{lemma}

\begin{proof}
  The beta integral gives
  \begin{equation}
    K_r(x,z)
    =r\int_0^1(1-u)^{r-1}u^{d_H(x,z)}\,du.
    \label{eq:hamming-beta-integral}
  \end{equation}
  For fixed $u\in[0,1]$, the kernel $u^{d_H(x,z)}$ is the $N$-fold tensor
  product of
  \begin{equation}
    \begin{pmatrix}1&u\\ u&1\end{pmatrix}.
    \label{eq:binary-hamming-kernel}
  \end{equation}
  This $2\times2$ matrix has eigenvalues $1+u$ and $1-u$, so it is positive
  semidefinite.  Tensor products, nonnegative integrals, and restrictions to
  principal submatrices preserve positive semidefiniteness, proving the
  claim.
\end{proof}

We will use the two special cases
\begin{equation}
  K_1(x,z)=\frac{1}{d_H(x,z)+1},
  \qquad
  K_2(x,z)=\binom{d_H(x,z)+2}{2}^{-1}.
  \label{eq:hamming-special-cases}
\end{equation}

\subsection{Type \texorpdfstring{$A_n$}{A-n}: the simplex quotient}
\label{subsec:type-a-certificate}

For an endpoint parabolic $A_{n-1}\subset A_n$, the quotient
$A_n/A_{n-1}$ is the set of $n+1$ vertices of a regular simplex.  Every pair
of distinct vertices is exchanged by a unique reflection, while all other
vertices lie in its wall.  Thus the positive side contains one vertex and
we take
\begin{equation}
  Q^{(\alpha)}=(1).
  \label{eq:type-a-quotient-matrix}
\end{equation}
The unique reflection that exchanges $x$ and $y$ supplies the sole nonzero
term in Eq.~\eqref{eq:vertex-convexity}.  The quotient is therefore
vertex-convex.

For $n\geq3$, the two endpoint parabolics of type $A_{n-1}$ cover all
generators of $A_n$.  Corollary~\ref{cor:parabolic-cover} and induction give
edge-convexity, starting from the one-edge case $A_1$, where there are no
distinct edges of the same label, and the
rank-two case $A_2=I_2(3)$.

\subsection{Type \texorpdfstring{$B_n$}{B-n}: the orthoplex and cube quotients}
\label{subsec:type-b-certificate}

Let $e_1,\ldots,e_n$ be the standard orthonormal basis of $\mathbb R^n$.
The roots of $B_n$ are the short roots $\pm e_i$ and the long roots
$\pm e_i\pm e_j$ for $i\neq j$.

First consider the quotient $B_n/B_{n-1}$.  It is the orthoplex
\begin{equation}
  X=\{\pm e_1,\ldots,\pm e_n\}.
  \label{eq:orthoplex-quotient}
\end{equation}
Any two distinct vertices are exchanged by a unique reflection.  Opposite
vertices are exchanged by a short-root reflection; all other pairs are
exchanged by the reflection perpendicular to their difference.  On every
positive side take $Q^{(\alpha)}$ to be the identity matrix.  Only the
diagonal entry associated with the exchanging reflection contributes to
Eq.~\eqref{eq:vertex-convexity}, so this quotient is vertex-convex.

The second quotient, $B_n/A_{n-1}$, is the sign cube $\{\pm1\}^n$.  For the
wall perpendicular to the short root $e_i$, its positive side consists of
the vertices with $x_i=1$.  Define
\begin{equation}
  Q^{(e_i)}_{x,z}:=\frac{1}{d_H(x,z)+1},
  \qquad x_i=z_i=1,
  \label{eq:type-b-cube-matrix}
\end{equation}
transport this matrix equivariantly to the other short-root walls, and use
the zero matrix on every long-root wall.  Equation~\eqref{eq:type-b-cube-matrix}
is a principal submatrix of $K_1$, hence is positive semidefinite by
Lemma~\ref{lem:hamming-kernel}.

It remains to check the normalization.  If $x,y\in\{\pm1\}^n$ differ in $k$
coordinates, exactly $k$ oriented short roots separate them.  Reflection in
one such wall corrects one differing coordinate, so
\begin{equation}
  d_H(x,s_\alpha y)=k-1,
  \qquad
  Q^{(\alpha)}_{x,s_\alpha y}=\frac1k.
  \label{eq:type-b-normalization}
\end{equation}
Their $k$ contributions sum to one.  The long-root walls contribute zero,
so Eq.~\eqref{eq:vertex-convexity} holds.

The standard parabolics $B_{n-1}$ and $A_{n-1}$ cover all generators of
$B_n$.  Their orthoplex and cube quotients are vertex-convex by the two
constructions above.  Corollary~\ref{cor:parabolic-cover} therefore proves
the result inductively, with $B_2=I_2(4)$ as the rank-two base case.

\subsection{Type \texorpdfstring{$D_n$}{D-n}: the demi-hypercube quotient}
\label{subsec:type-d-certificate}

The roots of $D_n$ are $\pm e_i\pm e_j$ with $i\neq j$.  For either endpoint
parabolic $A_{n-1}\subset D_n$, the quotient is one of the two parity classes
of sign vectors.  It is enough to treat
\begin{equation}
  X=\{x\in\{\pm1\}^n:x_1x_2\cdots x_n=1\},
  \label{eq:demihypercube-quotient}
\end{equation}
because the same matrix formula applies to the other parity class.  On the
positive side of every root wall define
\begin{equation}
  Q^{(\alpha)}_{x,z}:=
  \binom{d_H(x,z)+2}{2}^{-1}.
  \label{eq:type-d-matrix}
\end{equation}
This is a restriction of $K_2$, so it is positive semidefinite by
Lemma~\ref{lem:hamming-kernel}.

Two distinct vertices $x,y\in X$ differ in an even number $k$ of
coordinates.  A root wall strictly separates them precisely when its two
coordinates are chosen from these $k$ differing coordinates, with the root
oriented toward $x$.  Hence there are $\binom{k}{2}$ separating roots.
Reflection in any one of them corrects the chosen two coordinates and leaves
the others unchanged.  Therefore
\begin{equation}
  d_H(x,s_\alpha y)=k-2,
  \qquad
  Q^{(\alpha)}_{x,s_\alpha y}=\binom{k}{2}^{-1}.
  \label{eq:type-d-normalization}
\end{equation}
The $\binom{k}{2}$ terms sum to one, proving
Eq.~\eqref{eq:vertex-convexity}.

The two endpoint parabolics of type $A_{n-1}$ cover the two branches of the
$D_n$ diagram.  Type $A$ is already edge-convex, and both parity quotients
are vertex-convex by the argument above.  Corollary~\ref{cor:parabolic-cover}
thus proves edge-convexity of $D_n$ for $n\geq4$.

The quotient matrices used in the induction are summarized in
Table~\ref{tab:classical-quotients}.
\begin{table}[t]
  \centering
  \begin{tabular}{cccc}
    \toprule
    quotient & geometry & number of vertices & wall matrix \\
    \midrule
    $A_n/A_{n-1}$ & simplex & $n+1$ & $(1)$ \\
    $B_n/B_{n-1}$ & orthoplex & $2n$ & identity \\
    $B_n/A_{n-1}$ & cube & $2^n$ & $K_1$ restricted \\
    $D_n/A_{n-1}$ & demi-hypercube & $2^{n-1}$ & $K_2$ restricted\\
    \bottomrule
  \end{tabular}
  \caption{Analytic quotient matrices for the classical families.}
  \label{tab:classical-quotients}
\end{table}

The rank-two argument and the three classical inductions prove
Theorem~\ref{thm:classical-edge-convexity}.  It remains to treat the finite
exceptional Coxeter systems.  Most exceptional quotients admit similarly
small matrices; the two largest cases require exact symmetry reduction.

\section{Exceptional quotient constructions and completion of the proof}
\label{sec:exceptional-certificates}

Six irreducible Coxeter systems remain. Their quotient matrices vary in
size: the positive side of $E_8/D_7$ has $702$
vertices, and that of $H_4/A_3$ has $270$. The stabilizer of an oriented wall
reduces these matrices to blocks of order at most $14$, whose positivity
we check exactly.

\begin{theorem}[Exceptional systems]
  \label{thm:exceptional-edge-convexity}
  The standard Cayley graphs of the Coxeter systems
  \begin{equation}
    E_6,\qquad E_7,\qquad E_8,\qquad F_4,\qquad H_3,\qquad H_4
    \label{eq:exceptional-families}
  \end{equation}
  are edge-convex.
\end{theorem}

\subsection{The exceptional induction scheme}
\label{subsec:exceptional-induction}

Table~\ref{tab:exceptional-covers} gives the parabolic covers used below.
The two parabolics in each row cover every node of the corresponding Coxeter
diagram. The quotient sizes identify the two large matrix constructions.

\begin{table}[ht]
  \centering
  \begin{tabular}{ccc}
    \toprule
    system & covering parabolics & quotient sizes \\
    \midrule
    $E_6$ & $D_5,D_5$ & $27,27$ \\
    $E_7$ & $E_6,D_6$ & $56,126$ \\
    $E_8$ & $E_7,D_7$ & $240,2160$ \\
    $F_4$ & $B_3,B_3$ & $24,24$ \\
    $H_3$ & $I_2(5),A_2$ & $12,20$ \\
    $H_4$ & $H_3,A_3$ & $120,600$\\
    \bottomrule
  \end{tabular}
  \caption{Parabolic covers for the exceptional systems.}
  \label{tab:exceptional-covers}
\end{table}

All parabolics in the first row are classical.  Once $E_6$ is established,
the second row proves $E_7$, and the third then proves $E_8$.  The other
dependencies are $B_3\to F_4$, $I_2(5),A_2\to H_3$, and $H_3\to H_4$. By
Corollary~\ref{cor:parabolic-cover}, it remains to establish vertex-convexity
for the quotients in the table.

\subsection{The small quotients}
\label{subsec:small-exceptional-quotients}

We first treat every quotient except $E_8/D_7$ and $H_4/A_3$. Each
matrix below depends only on inner products and is therefore equivariant.

\begin{proposition}[Small exceptional quotients]
  \label{prop:small-exceptional-quotients}
  The quotients
  \begin{equation}
    \begin{gathered}
      E_6/D_5,\quad E_7/E_6,\quad E_7/D_6,\quad E_8/E_7,
      \quad F_4/B_3,\\
      H_3/I_2(5),\quad H_3/A_2,\quad H_4/H_3
    \end{gathered}
    \label{eq:small-exceptional-quotients}
  \end{equation}
  are vertex-convex.
\end{proposition}

\begin{proof}
  We use the standard coordinate realizations of the finite root systems
  \cite{Humphreys1990}. All orbit sizes, counts of vertices on each side of
  a wall, and spectra quoted below are obtained in exact arithmetic from
  those coordinates. The
  ancillary file \nolinkurl{notes/verify_small_quotients.py} reconstructs
  the matrices, checks every ordered-pair normalization, and verifies
  positive semidefiniteness by rational symmetric elimination. The listed
  spectra include every distinct eigenvalue; multiplicities do not affect
  positivity.

  Normalize the roots of $E_6,E_7$, and $E_8$ to have squared norm two.
  For $E_6/D_5$, the quotient is a $27$-point orbit.  Distinct quotient
  vertices have inner product $1/3$ or $-2/3$, and a positive wall side has
  six vertices.  On that side take
  \begin{equation}
    Q^{(\alpha)}=\frac78 I+\frac18J.
    \label{eq:e6-d5-matrix}
  \end{equation}
  Here and below, $I$ is the identity matrix and $J$ is the all-ones matrix
  on the stated index set.  The eigenvalues of
  Eq.~\eqref{eq:e6-d5-matrix} are $7/8$ and $13/8$.  A pair of vertices with
  inner product $1/3$ receives one diagonal contribution, while a pair with
  inner product $-2/3$ receives eight contributions of $1/8$.  This proves
  both positivity and Eq.~\eqref{eq:vertex-convexity}.

  The other endpoint $D_5$ parabolic gives the dual orbit, represented by
  the negatives of these $27$ vertices. Negation preserves inner products
  and exchanges the two sides of each wall, so the same construction
  establishes vertex-convexity for that quotient.

  The quotient $E_7/E_6$ is a $56$-point orbit.  Its vertices have squared
  norm $3/2$, and a positive wall side contains $12$ vertices.  Define
  \begin{equation}
    Q^{(\alpha)}_{x,z}=
    \begin{cases}
      1,    &(x,z)=3/2,\\
      1/10, &(x,z)=1/2,\\
      1/27, &(x,z)=-1/2.
    \end{cases}
    \label{eq:e7-e6-matrix}
  \end{equation}
  Its distinct eigenvalues are
  $55/27$, $26/27$, and $113/135$.  For pairs of quotient vertices with
  inner products $1/2,-1/2,-3/2$, respectively, the normalization is
  $1$, $10(1/10)$, and $27(1/27)$.

  The quotient $E_7/D_6$ is the set of $126$ roots of $E_7$.  For a fixed
  root $\alpha$, its positive side is $\{\alpha\}\sqcup Y$, where $Y$
  contains the $32$ roots $x$ satisfying $(x,\alpha)=1$.  Set
  $Q^{(\alpha)}_{\alpha,\alpha}=1$ and
  $Q^{(\alpha)}_{\alpha,x}=0$ for $x\in Y$.  On $Y$ take
  \begin{equation}
    Q^{(\alpha)}_{x,z}=
    \begin{cases}
      1,    &(x,z)=2,\\
      1/8,  &(x,z)=1,\\
      1/15, &(x,z)=0,\\
      0,    &(x,z)=-1.
    \end{cases}
    \label{eq:e7-d6-matrix}
  \end{equation}
  The four distinct eigenvalues of this block are
  $31/8$, $31/24$, $33/40$, and $97/120$.  The normalization contributions
  for pair inner products $1,0,-1,-2$ are, respectively,
  \begin{equation}
    1,\qquad 8\left(\frac18\right),\qquad
    15\left(\frac1{15}\right)+2(0),\qquad
    1+32(0).
    \label{eq:e7-d6-normalization}
  \end{equation}

  The quotient $E_8/E_7$ is the set of $240$ roots of $E_8$.  Its positive
  side is $\{\alpha\}\sqcup Y$, now with $|Y|=56$. Keep the same isolated
  entry at $\alpha$. On $Y$, define
  \begin{equation}
    Q^{(\alpha)}_{x,z}=
    \begin{cases}
      1,    &(x,z)=2,\\
      1/12, &(x,z)=1,\\
      1/27, &(x,z)=0,\\
      0,    &(x,z)=-1.
    \end{cases}
    \label{eq:e8-e7-matrix}
  \end{equation}
  The block on $Y$ has distinct eigenvalues
  $31/36$, $95/108$, $17/12$, and $17/4$.  The four normalization
  contributions are
  \begin{equation}
    1,\qquad 12\left(\frac1{12}\right),\qquad
    27\left(\frac1{27}\right)+2(0),\qquad
    1+56(0).
    \label{eq:e8-e7-normalization}
  \end{equation}

  The roots of $F_4$ have two lengths, called short and long.  Realize
  $F_4/B_3$ as the $24$ long roots
  $\{\pm e_i\pm e_j:i<j\}$.  There are short-root and long-root walls.  Use
  the identity matrix for a short-root wall.  For a long root $\alpha$, the
  positive side is $\{\alpha\}\sqcup Y$ with $|Y|=8$. Set
  $Q^{(\alpha)}_{\alpha,\alpha}=1$, set the cross entries
  from $\alpha$ to zero, and define the block on $Y$ by
  \begin{equation}
    Q^{(\alpha)}_{x,z}=
    \begin{cases}
      1,   &x=z,\\
      1/3, &x\neq z\ \text{and}\ (x,z)=0,\\
      0,   &\text{otherwise}.
    \end{cases}
    \label{eq:f4-b3-matrix}
  \end{equation}
  The graph on $Y$ joining orthogonal pairs consists of two disjoint
  complete graphs on four vertices. In a complete graph, every pair of
  distinct vertices is adjacent. Hence the eigenvalues are $1$, $2$, and $2/3$, all
  nonnegative.  The nontrivial normalization is $3(1/3)=1$; all remaining
  pair classes receive one diagonal contribution. The other endpoint
  parabolic gives the short-root orbit. The orthogonal map
  \[
    U(x_1,x_2,x_3,x_4)
    =\frac1{\sqrt2}(x_1+x_2,x_1-x_2,x_3+x_4,x_3-x_4)
  \]
  exchanges the long- and short-root directions and preserves the reflection
  arrangement. Transporting the construction by $U$ and rescaling the orbit
  gives the second quotient matrix, with the two wall types exchanged.

  For a fixed root $\alpha$, call $(x,\alpha)$ the \emph{height} of a
  quotient vertex $x$.  The quotient $H_3/I_2(5)$ has $12$ vertices and four
  vertices on a positive wall side. Order its two height layers by increasing
  height and take
  \begin{equation}
    Q^{(\alpha)}=I_2\oplus
    \begin{pmatrix}1&1/5\\[2pt]1/5&1\end{pmatrix}.
    \label{eq:h3-i25-matrix}
  \end{equation}
  Its eigenvalues are $1,1,6/5,4/5$, and the only non-diagonal normalization
  is $5(1/5)=1$.  Similarly, $H_3/A_2$ has $20$ vertices and positive-side
  height layers of sizes $2,4,2$, again in increasing order. On the middle
  layer, join two distinct vertices when their inner product is one third
  of their common squared norm. These pairs form a four-cycle $C_4$.
  Writing $A(C_4)$ for its adjacency matrix, take
  \begin{equation}
    Q^{(\alpha)}=I_2\oplus
    \left(I_4+\frac12A(C_4)\right)\oplus
    \begin{pmatrix}1&1/3\\[2pt]1/3&1\end{pmatrix}.
    \label{eq:h3-a2-matrix}
  \end{equation}
  The three blocks have eigenvalues
  $1,1$; $2,1,1,0$; and $4/3,2/3$.  Their non-diagonal normalization
  equations are $2(1/2)=1$ and $3(1/3)=1$.

  It remains to verify $H_4/H_3$.  Put
  \begin{equation}
    \phi:=\frac{1+\sqrt5}{2},\qquad
    a:=\frac1{2\phi},\qquad b:=\frac12,\qquad c:=\frac\phi2.
    \label{eq:h4-heights}
  \end{equation}
  The quotient is the $120$-point root orbit, with roots normalized to have
  unit norm.  A positive wall side has height layers
  \begin{equation}
    \begin{array}{ccccc}
      h=(x,\alpha)&1&c&b&a\\ \hline
      \#\{x:(x,\alpha)=h\}&1&12&20&12.
    \end{array}
    \label{eq:h4-h3-height-counts}
  \end{equation}
  Make $Q^{(\alpha)}$ block diagonal by height.  For vertices $x,z$ in the
  height-$h$ block, let $q_h(u)$ denote its entry when $(x,z)=u$.  Its
  nonzero values are
  \begin{equation}
    \begin{array}{ccc@{\qquad}ccc}
      h&u&q_h(u)&h&u&q_h(u)\\ \hline
      1&1&1/689&a&1&1\\
      a&b&1/10&a&-a&1/18\\
      a&-c&12/689&b&1&1\\
      b&c&1/10&b&b&1/4\\
      b&0&1/9&b&-a&1/10\\
      b&-b&20/689&c&1&1\\
      c&c&1/18&c&b&1/10\\
      c&a&12/689&&&
    \end{array}
    \label{eq:h4-h3-coefficients}
  \end{equation}
  In a row of the height-$a$ block, the off-diagonal entries with inner
  products $b,-a,-c$ occur $5,5,1$ times, respectively; in a row of the height-$c$
  block, those with inner products $c,b,a$ have the same counts. Thus the
  absolute off-diagonal row sum in either block is
  \begin{equation}
    5\left(\frac1{10}\right)
    +5\left(\frac1{18}\right)+\frac{12}{689}
    =\frac{4931}{6201}<1.
    \label{eq:h4-h3-diagonal-dominance}
  \end{equation}
  Since every diagonal entry is one, both blocks are positive definite by
  diagonal dominance. The characteristic
  polynomial of the height-$b$ block is
  \begin{equation}
    \begin{split}
      &\left(\lambda-\frac{4583}{8268}\right)^4
       \left(\lambda-\frac{15487}{12402}\right)^6
       \left(\lambda-\frac{78457}{20670}\right)\\
      &\qquad\times
       \left(\lambda-\frac{31427}{62010}\right)^5
       \left(\lambda-\frac{122797}{124020}\right)^4,
    \end{split}
    \label{eq:h4-h3-characteristic-polynomial}
  \end{equation}
  so every block is positive semidefinite. The normalization contributions
  for each pair inner product are
  \begin{equation}
    \begin{array}{cc}
      (x,y)&\text{contribution}\\ \hline
      c&q_a(1)\\
      b&q_b(1)\\
      a&5q_a(b)+5q_b(c)\\
      0&4q_b(b)\\
      -a&q_c(1)\\
      -b&3q_a(-a)+6q_b(0)+3q_c(c)\\
      -c&5q_b(-a)+5q_c(b)\\
      -1&12q_a(-c)+20q_b(-b)+12q_c(a)+q_1(1).
    \end{array}
    \label{eq:h4-h3-normalization}
  \end{equation}
  Substitution from Eq.~\eqref{eq:h4-h3-coefficients} makes every entry in
  the right column equal to one; the last numerator is
  $144+400+144+1=689$.  This completes all cases in
  Eq.~\eqref{eq:small-exceptional-quotients}.
\end{proof}

\subsection{The exact construction for \texorpdfstring{$E_8/D_7$}{E8/D7}}
\label{subsec:e8-d7-certificate}

For $E_8/D_7$, each irreducible representation of the wall stabilizer
occurs at most twice in the permutation representation on a positive wall
side. This reduces positivity of the $702$-dimensional
positive-side matrix to ten exact tests.

Let $L_{E_8}$ be the lattice
\begin{equation}
  \begin{split}
    L_{E_8}:={}&\left\{z\in\mathbb Z^8:\sum_i z_i\in2\mathbb Z\right\}\\
    &{}\cup
    \left\{z\in(\mathbb Z+\tfrac12)^8:\sum_i z_i\in2\mathbb Z\right\}.
  \end{split}
  \label{eq:e8-lattice}
\end{equation}
The quotient consists of the lattice vectors of squared norm four:
\begin{equation}
  X=\{x\in L_{E_8}:(x,x)=4\},
  \qquad |X|=2160.
  \label{eq:e8-d7-shell}
\end{equation}
Fix $\alpha=e_1+e_2$.  Its positive side splits into the two height layers
\begin{equation}
  X^+(\alpha)=X_1\sqcup X_2,
  \qquad
  X_h:=\{x\in X:(x,\alpha)=h\},
  \qquad |X_1|=576,\qquad |X_2|=126.
  \label{eq:e8-d7-height-layers}
\end{equation}

For $x\in X_h$ and $z\in X_{h'}$, define
\begin{equation}
  Q^{(\alpha)}_{x,z}:=c_{hh';u},
  \qquad u:=(x,z).
  \label{eq:e8-d7-matrix}
\end{equation}
The coefficients are symmetric in $h,h'$; unlisted entries in
Table~\ref{tab:e8-d7-coefficients} are zero.

\begin{table}[ht]
  \centering
  \small
  \begin{tabular}{cccc}
    \toprule
    $(h,h';u)$ & $c_{hh';u}$ & $(h,h';u)$ & $c_{hh';u}$\\
    \midrule
    $(1,1;4)$  & $1$            & $(1,2;2)$  & $13/112$\\
    $(1,1;3)$  & $1/8$          & $(1,2;1)$  & $4/63$\\
    $(1,1;2)$  & $37/900$       & $(1,2;0)$  & $1/68$\\
    $(1,1;1)$  & $1/224$        & $(1,2;-1)$ & $-1/22$\\
    $(1,1;0)$  & $-109/13230$   & $(2,2;4)$  & $6/7$\\
    $(1,1;-1)$ & $43/28560$     & $(2,2;3)$  & $37/90$\\
    $(1,1;-2)$ & $373/9625$     & $(2,2;2)$  & $17/70$\\
    $(1,1;-3)$ & $13/1728$      & $(2,2;1)$  & $1/25$\\
    $(1,2;3)$  & $23/120$       & $(2,2;0)$  & $1/27$\\
    \bottomrule
  \end{tabular}
  \caption{Exact coefficients for the $E_8/D_7$ quotient matrix.}
  \label{tab:e8-d7-coefficients}
\end{table}

There are nine orbits of ordered pairs of distinct quotient vertices. For a
fixed first vertex, the number of possible second vertices in each orbit is
called the orbit's subdegree. The inner products and subdegrees are
\begin{equation}
  (3,64),\ (2,280),\ (1,448),\ (0,560),\ (0,14),\
  (-1,448),\ (-2,280),\ (-3,64),\ (-4,1).
  \label{eq:e8-d7-pair-orbits}
\end{equation}
The two inner-product-zero orbits are thus distinguished by their
subdegrees. In the order displayed in
Eq.~\eqref{eq:e8-d7-pair-orbits}, the vertex-convexity equations are
\begin{align}
  c_{11;4}&=1,&
  8c_{11;3}&=1,\nonumber\\
  15c_{11;2}+2c_{12;3}&=1,&
  16c_{11;1}+8c_{12;2}&=1,\nonumber\\
  32c_{11;1}+c_{22;4}&=1,&
  21c_{11;0}+12c_{12;1}+c_{22;3}&=1,\nonumber\\
  24c_{11;-1}+16c_{12;0}+3c_{22;2}&=1,&
  35c_{11;-2}+14c_{12;-1}+7c_{22;1}&=1,\nonumber\\
  64c_{11;-3}+14c_{22;0}&=1.&&
  \label{eq:e8-d7-normalization}
\end{align}
The coefficients in Table~\ref{tab:e8-d7-coefficients} satisfy these
identities over the rational numbers.

It remains to prove positivity. Let $\mathbb R[X_h]$ denote the real-valued
functions on $X_h$, and write $\rho_d^{(h)}$ for the irreducible constituent of dimension
$d$ in the height-$h$ permutation module below. The superscript $h$
identifies the height layer. The oriented-wall
stabilizer has Coxeter type $E_7$ and acts
transitively on both height layers. The point stabilizers have ten and five
orbits, respectively. For each layer, the corresponding invariant matrix
algebra therefore has dimension ten or five. Restricting an invariant matrix
to functions constant on the orbits of a point stabilizer gives a smaller
matrix. For the wall matrices on the two layers, these reduced matrices
have that many distinct eigenvalues. Powers of the wall matrices
consequently span the whole invariant algebra, which is
commutative. This proves that the permutation modules are multiplicity-free,
meaning that every irreducible constituent occurs once.
The ranks of the resulting spectral projectors give the decompositions
\begin{align}
  \mathbb R[X_1]&\cong
  \rho^{(1)}_1\oplus\rho^{(1)}_7\oplus\rho^{(1)}_{15}
  \oplus\rho^{(1)}_{21}\oplus\rho^{(1)}_{27}
  \oplus\rho^{(1)}_{35}\oplus\rho^{(1)}_{56}
  \oplus\rho^{(1)}_{105}\oplus\rho^{(1)}_{120}
  \oplus\rho^{(1)}_{189},
  \label{eq:e8-x1-decomposition}\\
  \mathbb R[X_2]&\cong
  \rho^{(2)}_1\oplus\rho^{(2)}_7\oplus\rho^{(2)}_{27}
  \oplus\rho^{(2)}_{35}\oplus\rho^{(2)}_{56}.
  \label{eq:e8-x2-decomposition}
\end{align}
Write the matrix in height blocks as
\begin{equation}
  Q^{(\alpha)}=
  \begin{pmatrix}A&B\\ B^{\mathsf T}&C\end{pmatrix}.
  \label{eq:e8-d7-height-blocks}
\end{equation}
The reduced matrices for $A$ and $C$ have orders ten and five and give their
spectra exactly.  The five constituents occurring only in $X_1$ have the
positive eigenvalues listed in the last column of
Table~\ref{tab:e8-d7-positivity}.

For $d\in\{1,7,27,35,56\}$, let $a_d$ and $c_d$ be the scalars by which
$A$ and $C$ act on $\rho_d^{(1)}$ and $\rho_d^{(2)}$, and let $k_d$ be the
scalar by which $B^{\mathsf T}B$ acts on $\rho_d^{(2)}$. The values in
Table~\ref{tab:e8-d7-positivity} satisfy $k_d>0$, so $B$ pairs the two
constituents and proves that they are isomorphic. After an orthonormal
identification, Schur's lemma reduces Eq.~\eqref{eq:e8-d7-height-blocks} to
\begin{equation}
  \begin{pmatrix}a_d&\sqrt{k_d}\\ \sqrt{k_d}&c_d\end{pmatrix}.
  \label{eq:e8-d7-two-by-two-block}
\end{equation}
\begin{table}[H]
  \centering
  \scriptsize
  \begin{tabular}{ccccc}
    \toprule
    common $d$ & $a_d$ & $c_d$ & $X_1$-only $d$ & eigenvalue\\
    \midrule
    $1$  & $86105893/8078400$ & $141283/4725$
         & $15$  & $73122827/94248000$\\
    $7$  & $58784281/11309760$ & $31931/4725$
         & $21$  & $131256953/131947200$\\
    $27$ & $212406979/56548800$ & $1276/4725$
         & $105$ & $78004567/131947200$\\
    $35$ & $1394542169/1979208000$ & $2584/4725$
         & $120$ & $23917127/31416000$\\
    $56$ & $2076599863/1979208000$ & $368/4725$
         & $189$ & $1656315623/1979208000$\\
    \bottomrule
  \end{tabular}

  \medskip
  \begin{tabular}{ccc}
    \toprule
    common $d$ & $k_d$ & $a_d c_d-k_d$\\
    \midrule
    $1$  & $137669939521/440609400$
         & $6379695779179/1019696040000$\\
    $7$  & $94995253369/3084265800$
         & $43222023010697/9993021192000$\\
    $27$ & $284089/13721400$
         & $730123123361/734780970000$\\
    $35$ & $8281/62944200$
         & $42101491335947/109298669287500$\\
    $56$ & $7295401/308426580$
         & $1586537070347/27324667321875$\\
    \bottomrule
  \end{tabular}
  \caption{Exact positivity data for $E_8/D_7$.}
  \label{tab:e8-d7-positivity}
\end{table}

Every displayed number is positive. The $2\times2$ blocks in
Eq.~\eqref{eq:e8-d7-two-by-two-block} have positive leading entries and
positive determinants, while all remaining constituents carry the positive
scalar eigenvalues in Table~\ref{tab:e8-d7-positivity}.  Hence
$Q^{(\alpha)}$ is positive definite, meaning that its quadratic form is
strictly positive on every nonzero vector.  Equations
\eqref{eq:e8-d7-normalization} and \eqref{eq:quotient-equivariance} now prove
that $E_8/D_7$ is vertex-convex.

\subsection{The exact construction for \texorpdfstring{$H_4/A_3$}{H4/A3}}
\label{subsec:h4-a3-certificate}

The final quotient has $600$ vertices and $270$ vertices on a positive wall
side. We construct its matrix from inner products and block-diagonalize it
using the wall stabilizer.

Fix a root $\alpha$.  A \emph{Gram class} is a set of pairs $(x,z)$ for which
the inner product and the unordered pair of heights $(x,\alpha)$ and
$(z,\alpha)$ are fixed.  On
$X^+(\alpha)\times X^+(\alpha)$ there are $317$ such classes.
Each is a union of the $393$ symmetrized orbitals of the oriented wall
stabilizer, where a symmetrized orbital is an orbital together with its
transpose. These classes are symmetric and invariant, so
Lemma~\ref{lem:invariant-class-certificate} applies.

Let $A_1,\ldots,A_{317}$ be their indicator matrices, and let $N$ be the
$35\times317$ count matrix in Eq.~\eqref{eq:invariant-class-count}, with one
row for each orbit of unordered pairs of distinct quotient vertices. Unordered
pairs suffice because reversing $(x,y)$ sends an oriented separating root
$\beta$ to $-\beta$, while $s_{-\beta}=s_\beta$ and the matrices are
symmetric.

Let $D$ be diagonal, with $D_{ii}$ the number of ordered pairs in
the $i$th Gram class, and let $c^{(0)}$ encode the identity matrix. Exact
enumeration shows that $N$ has full row rank. Set
\begin{equation}
  c:=c^{(0)}-D^{-1}N^{\mathsf T}
  \bigl(ND^{-1}N^{\mathsf T}\bigr)^{-1}
  \bigl(Nc^{(0)}-\boldsymbol1\bigr),
  \qquad
  Q^{(\alpha)}:=\sum_{i=1}^{317}c_iA_i.
  \label{eq:h4-a3-projected-certificate}
\end{equation}
All entries of $N,D$, and $c$ are rational.  Multiplying the first equation
by $N$ gives $Nc=\boldsymbol1$ exactly, so only positivity remains.

Let $H$ be the oriented-wall stabilizer.  It is isomorphic to
\begin{equation}
  H\cong\operatorname{Alt}(5)\times C_2,
  \label{eq:h4-wall-stabilizer}
\end{equation}
where $\operatorname{Alt}(5)$ is the alternating group on five letters and
$C_2$ is the cyclic group of order two.  Write $\mathbf d_\epsilon$ for a
real irreducible representation of $\operatorname{Alt}(5)$ of dimension
$d$, with central sign $\epsilon\in\{+,-\}$ specifying the action of the
nontrivial element of the cyclic factor. The permutation representation
on the positive wall side decomposes as
\begin{equation}
  \begin{split}
    \mathbb R[X^+(\alpha)]\cong{}&
    7\!\cdot\!\mathbf1_+\oplus3\!\cdot\!\mathbf3_+
    \oplus3\!\cdot\!\mathbf3'_+
    \oplus10\!\cdot\!\mathbf3_-\oplus10\!\cdot\!\mathbf3'_-\\
    &{}\oplus10\!\cdot\!\mathbf4_+\oplus10\!\cdot\!\mathbf4_-
    \oplus14\!\cdot\!\mathbf5_+\oplus7\!\cdot\!\mathbf5_-.
  \end{split}
  \label{eq:h4-a3-decomposition}
\end{equation}
The prime distinguishes the two inequivalent three-dimensional
representations.

For each irreducible representation, averaging the group action with its
character weights projects onto its
\emph{isotypic component}, meaning the sum of all copies of that irreducible
representation. This projection is followed by averaging over a subgroup
of the rotational $\operatorname{Alt}(5)$ factor whose fixed subspace in the
irreducible representation is one-dimensional. We use an order-two subgroup
for $\mathbf3$ and $\mathbf3'$, a subgroup $C_2\times C_2$ for $\mathbf4$,
and an order-three subgroup for $\mathbf5$. The same choices work for both
central signs. The trivial representation requires no further averaging.
Choose independent projected coordinate vectors as columns of a matrix
$U_\rho$. In this generally nonorthonormal basis the restricted quadratic
form is $H_\rho=U_\rho^{\mathsf T}Q^{(\alpha)}U_\rho$, with positive-definite
Gram matrix $G_\rho=U_\rho^{\mathsf T}U_\rho$.
Changing to an orthonormal basis gives the block
$G_\rho^{-1/2}H_\rho G_\rho^{-1/2}$. It is positive definite exactly when
$H_\rho$ is. We test these nine restricted forms, whose orders appear
in Table~\ref{tab:h4-a3-blocks}. Their entries lie in
$\mathbb Q(\sqrt5)$, the field of numbers $p+q\sqrt5$ with rational
$p$ and $q$.

The quotient matrix has the positive common denominator
\begin{equation}
  D_Q=626605405826400.
  \label{eq:h4-a3-common-denominator}
\end{equation}
In the stored bases, exact
$LDL^{\mathsf T}$ elimination gives the pivot bounds in the last column of
Table~\ref{tab:h4-a3-blocks}. The pivots are the diagonal entries of the
diagonal factor for $H_\rho$. The compact data store $D_QH_\rho$;
the verifier divides the resulting pivots by $D_Q$ before testing these
bounds. The rational enclosure
\begin{equation}
  \frac{809}{500}<\phi<\frac{1619}{1000}
  \label{eq:golden-ratio-enclosure}
\end{equation}
is sufficient to verify every inequality.

\begin{table}[H]
  \centering
  \begin{tabular}{cccc}
    \toprule
    component & block order & fixed-line subgroup & every pivot is greater than\\
    \midrule
    $\mathbf1_+$  & $7$  & trivial & $18$\\
    $\mathbf3_+$  & $3$  & order two & $28$\\
    $\mathbf3'_+$ & $3$  & order two & $72$\\
    $\mathbf3_-$  & $10$ & order two & $21$\\
    $\mathbf3'_-$ & $10$ & order two & $6$\\
    $\mathbf4_+$  & $10$ & $C_2\times C_2$ & $97$\\
    $\mathbf4_-$  & $10$ & $C_2\times C_2$ & $6$\\
    $\mathbf5_+$  & $14$ & order three & $19$\\
    $\mathbf5_-$  & $7$  & order three & $11$\\
    \bottomrule
  \end{tabular}
  \caption{Exact positivity tests for the restricted multiplicity forms
  $H_\rho$ of $H_4/A_3$, in the bases specified by the ancillary construction.}
  \label{tab:h4-a3-blocks}
\end{table}

Every pivot is positive, so every multiplicity block is positive definite.
Schur's lemma and Eq.~\eqref{eq:h4-a3-decomposition} then show that
$Q^{(\alpha)}$ is positive definite.  Together with the exact identity
$Nc=\boldsymbol1$, Lemma~\ref{lem:invariant-class-certificate} proves that
$H_4/A_3$ is vertex-convex.

\subsection{Ancillary data and exact verification}
\label{subsec:ancillary-certificates}

The small-quotient verifier checks both $F_4/B_3$ embeddings and all other
cases of Proposition~\ref{prop:small-exceptional-quotients} directly. The
two large quotient arguments are exact computer-assisted proofs. For
$E_8/D_7$, the file
\nolinkurl{notes/e8_d7_exact_certificate.py} reconstructs the norm-four
shell, the pair-orbit equations, the reduced matrices, and every entry of
Table~\ref{tab:e8-d7-positivity} over $\mathbb Q$. For $H_4/A_3$, the file
\nolinkurl{h4_a3_certificate_data.json} stores the $317$ Gram classes, the
$35\times317$ integer matrix $N$, and the rational coefficient vector, while
\nolinkurl{h4_a3_commutant_blocks_integer.json} stores the nine exact
multiplicity forms. The construction mode of
\nolinkurl{h4_a3_commutant_blocks.py} reconstructs the geometry, verifies
$Nc=\boldsymbol1$, and forms the character projections, restricted forms,
and Gram matrices. Exact $LDL^{\mathsf T}$ tests verify their positivity;
the Gram tests also verify independence of the selected basis vectors.
The script's \texttt{--verify-compact} mode checks only positivity of the stored
forms, not their reconstruction. An independent check in
\nolinkurl{h4_a3_exact_certificate.py} reconstructs $Q^{(\alpha)}$ and
verifies a rational residual bound using the supplied dyadic factor. The accompanying manifest
\nolinkurl{ANCILLARY_CERTIFICATES.md} gives the verification commands and
SHA-256 digests. In the block construction, floating-point arithmetic selects candidate
basis columns; exact arithmetic then verifies their independence and tests
positivity. Decimal eigenvalues are diagnostic only.

\subsection{Completion of the classification}
\label{subsec:classification-completion}

Apply the covers in Table~\ref{tab:exceptional-covers}. The two
$D_5$ quotients and Proposition~\ref{prop:small-exceptional-quotients} prove
$E_6$ edge-convex. The $E_6$ and $D_6$ quotients then give edge-convexity
of $E_7$, and the $E_7$ quotient together with the exact $E_8/D_7$
construction gives edge-convexity of $E_8$. The two $B_3$ quotients treat
$F_4$. The $I_2(5)$ and $A_2$ quotients establish edge-convexity of $H_3$,
after which the $H_3$ quotient and the exact $H_4/A_3$ construction give
edge-convexity of $H_4$. This proves Theorem~\ref{thm:exceptional-edge-convexity}.

Theorem~\ref{thm:classical-edge-convexity} and
Theorem~\ref{thm:exceptional-edge-convexity} cover every irreducible finite
Coxeter system.  Proposition~\ref{prop:coxeter-products} covers disconnected
Coxeter diagrams, so every finite Coxeter Cayley graph is edge-convex.
Conversely, Section~\ref{sec:necessity} showed that every connected
edge-convex $\psi$-graph reduces to such a Cayley graph after merging
globally parallel colours, and that restoring those colours preserves
edge-convexity. This completes the proof of
Theorem~\ref{thm:classification}, and Corollary~\ref{cor:coxeter-monotones}
then follows from the edge-convexity criterion.

\section{Conclusion and outlook}
\label{sec:discussion}

Theorem~\ref{thm:classification} classifies connected edge-convex
$\psi$-graphs. After merging globally parallel colour families as in
Section~\ref{subsec:mirror-coxeter}, these are precisely the standard
edge-labelled Cayley graphs of finite Coxeter systems. By
Corollary~\ref{cor:coxeter-monotones}, every such graph gives a normalized
multi-invariant whose complement is a pure-state entanglement monotone.

The necessity proof uses the normalization conditions to show that
reflecting cuts separate every pair of equally labelled edges. The
sufficiency proof constructs compatible positive-semidefinite matrices on
these cuts. In the parabolic induction, quotient matrices control pairs of
edges in distinct fibres, while the lower-rank Coxeter graph controls pairs
within a fibre. The classical constructions are analytic. For the exceptional
quotients, normalization and positivity are verified by exact finite arithmetic.

Edge-convexity implies that $\widehat{\mathcal Z}$ is convex as a function
of $\rho_{\bar A}$ for every party $A$. We conjecture that the converse
also holds.

\newtheorem{conjecture}[theorem]{Conjecture}
\begin{conjecture}[Convexity implies edge-convexity]
  \label{conj:convexity-implies-edge-convexity}
  Let $\mathcal Z$ be a multi-invariant associated with a connected
  $\psi$-graph, and suppose that it is real and nonnegative for every finite
  choice of local Hilbert spaces. If its normalized form
  $\widehat{\mathcal Z}=\mathcal Z^{1/n_{\mathcal Z}}$ is convex as a
  function of $\rho_{\bar A}$ for every party $A$ and every such choice of
  local Hilbert spaces, then $\mathcal Z$ is edge-convex.
\end{conjecture}

The convexity hypothesis applies to every party and every finite choice
of local dimensions, with no bound on the dimension of the purifying
party. With these quantifiers, convexity is equivalent to LOCC
monotonicity of $1-\widehat{\mathcal Z}$. A local operation on $A$ produces
conditional reduced states satisfying
$\rho_{\bar A}=\sum_i p_i\rho_{\bar A,i}$, so convexity gives the
monotonicity inequality~\eqref{eq:psem}. Conversely, any decomposition of
$\rho_{\bar A}$ into density matrices can be realized by a local operation
on a sufficiently large purifying system. Monotonicity therefore implies
convexity, by the purification argument underlying the concavity criterion
for entanglement monotones~\cite{Vidal2000}. To realize the ensemble, apply
the Hughston--Jozsa--Wootters construction~\cite{HughstonJozsaWootters1993}
to a pure-state refinement of each density-matrix component and retain
the refinement label in the purifying system.

If Conjecture~\ref{conj:convexity-implies-edge-convexity} holds, the Coxeter
classification would give a complete classification of dimension-independent
pure-state entanglement monotones of the form
$1-\mathcal Z^{1/n_{\mathcal Z}}$, where $\mathcal Z$ is a connected
nonnegative multi-invariant. Here dimension-independent means that the same
contraction pattern defines a monotone for every finite choice of local
Hilbert spaces.

The convexity problem has a classical analog studied in the graph-norm
literature. For a tripartite multi-invariant, take a reduced
density matrix that is diagonal in a product basis,
\begin{equation}
  \rho_{BC}
  =\sum_{\beta,\gamma}p_{\beta\gamma}
  |\beta\gamma\rangle\langle\beta\gamma|.
\end{equation}
Associate to its $\psi$-graph a bipartite multigraph $H$: its left vertices are
the $AB$-cycles, its right vertices are the $AC$-cycles, and each $A$-edge
joins the two cycles containing it. Distinct $A$-edges can join the same
two cycles, so parallel edges must in general be retained. The contraction becomes the
weighted graph count
\begin{equation}
  \mathcal Z(\rho_{BC})
  =t(H,p)
  :=\sum_{\{\beta_u\},\{\gamma_v\}}
  \prod_{(u,v)\in E(H)}p_{\beta_u\gamma_v}.
\end{equation}
Here the sum runs over assignments of a $B$-index to each left vertex and a
$C$-index to each right vertex. Dividing this count by the number of index
assignments gives the homomorphism density used in graph-norm theory. Since
$H$ has $n_{\mathcal Z}$ edges, the normalized multi-invariant corresponds
to the root $t(H,p)^{1/|E(H)|}$. At fixed index dimensions, the count and
density differ by a positive constant, which does not affect convexity.
The root is homogeneous of degree one, so its convexity on probability
tables is equivalent to convexity on all nonnegative tables of those
dimensions.

Graph-norm theory asks when such roots define norms on functions. A graph
is called \emph{weakly norming} when the functional
$f\mapsto t(H,|f|)^{1/|E(H)|}$ is a norm; Hatami developed a systematic
approach to this question through graph versions of H\"older's
inequality~\cite{Hatami2010}. Lee and Sch\"ulke proved that, for simple
graphs, weak norming is equivalent to convexity of the homomorphism density
on nonnegative symmetric kernels~\cite{LeeSchulke2021}. Our count uses a
fixed bipartition and possibly different index sets for its two parts,
giving a graph-count problem on two spaces.

Conlon and Lee proved that reflection graphs, constructed from cosets of
parabolic subgroups of finite reflection groups, are weakly
norming~\cite{ConlonLee2017}. For a simple tripartite Coxeter $\psi$-graph,
the $AB$- and $AC$-cycles are the cosets of $W_{\{A,B\}}$ and
$W_{\{A,C\}}$, where the party labels denote the simple generators.
Their intersection, when nonempty, is a coset of
$W_{\{A\}}=\langle A\rangle$ and hence is exactly one $A$-edge.
Thus $H$ is precisely their parabolic-coset reflection graph and is simple.
Their proof uses repeated applications of
the Cauchy--Schwarz inequality across cuts associated with reflections.
In our proof, reflecting cuts instead express the required quadratic forms
as sums of nonnegative terms. Both arguments use reflection symmetry to
establish positivity.

The graph $H$ records the incidence of two-colour cycles in the $\psi$-graph,
and the diagonal restriction does not test variations involving off-diagonal
matrix elements of $\rho_{BC}$. Graph-norm results therefore do not by
themselves establish convexity of the quantum multi-invariant. The quantum
problem extends graph-count convexity from nonnegative probability tables
to positive-semidefinite matrices. Determining which graph-norm arguments
extend to these matrices may help resolve the converse conjecture.

\section*{Acknowledgements}

We would like to thank Jonathan Harper, Vineeth Krishna, Harshal Kulkarni,
Gautam Mandal, Shiraz Minwalla, Onkar Parrikar, Trakshu Sharma,
Piyush Shrivastava, Sandip Trivedi for interesting discussions.

We used OpenAI's GPT-6 Astra through Codex to assist with mathematical
calculations, development of verification code, and manuscript preparation.
The authors take full responsibility for the results and the final manuscript.

\paragraph{Funding information}
This work is supported by the Infosys Endowment for the study of the Quantum
Structure of Spacetime and by the SERB Ramanujan fellowship. We acknowledge
the support of the Department of Atomic Energy, Government of India, under
Project Identification No.~RTI 4002. Finally, we acknowledge our debt to the
people of India for their steady support to the study of the basic sciences.

\end{document}